\documentclass[reqno,a4paper,12pt]{amsart}

\usepackage[a4paper,hmargin=2cm,tmargin=3cm,bmargin=3cm]{geometry}
\usepackage[utf8x]{inputenc}
\usepackage[T1]{fontenc}
\usepackage{mathptmx}
\usepackage{amsmath,amssymb,amstext,amsthm,amscd,mathrsfs,eucal,bm,slashed}
\usepackage{graphicx,color}
\usepackage{array}
\usepackage{cite}
\usepackage{hyperref}
\usepackage{amsfonts}
\usepackage{amscd}
\usepackage{verbatim} 
\usepackage{fancyhdr}
\usepackage{tikz}
\usepackage{tikz-cd}
\usepackage{slashed}
\usetikzlibrary{matrix}
\usetikzlibrary{positioning}
\usepackage[all]{xy}

\numberwithin{equation}{section}

\theoremstyle{plain}
\newtheorem{thm}{Theorem}[section]
\newtheorem{prop}[thm]{Proposition}
\newtheorem{lem}[thm]{Lemma}

\newtheorem{remark}[thm]{Remark}

\newcommand{\la}{\langle}
\newcommand{\ra}{\rangle}
\renewcommand{\a}{\alpha}
\renewcommand{\b}{\beta}
\renewcommand{\c}{\nabla}

\newcommand{\G}{\Gamma}
\newcommand{\C}{\mathscr{C}}
\newcommand{\e}{\epsilon}
\newcommand{\ve}{\varepsilon}
\newcommand{\vp}{\varphi}
\renewcommand{\l}{\lambda}
\renewcommand{\L}{\Lambda}
\renewcommand{\k}{\kappa}
\newcommand{\m}{\mu}
\newcommand{\n}{\nu}
\renewcommand{\d}{\delta}

\newcommand{\w}{\omega}
\newcommand{\Om}{\Omega}
\renewcommand{\t}{\tau}
\newcommand{\T}{\mathscr{T}}
\newcommand{\U}{\Upsilon}

\newcommand{\p}{\partial}
\newcommand{\M}{\mathscr{M}}
\newcommand{\mc}{\mathcal}
\newcommand{\mf}{\mathfrak}
\newcommand{\mbb}{\mathbb}

\newcommand{\mr}{\mathring}

\newcommand{\wt}{\widetilde}
\newcommand{\wh}{\widehat}
\DeclareMathOperator{\tho}{\text{\rm\th}}
\DeclareMathOperator{\edt}{\text{\rm\dh}}
\newcommand{\PTc}{\mathbb{P}\mathcal{T}}
\newcommand{\CP}{\mbb{C}{\rm P}}
\newcommand{\SL}{{\rm SL}(2,\mathbb{C})}

\begin{document}

\title{Two-dimensional twistor manifolds and Teukolsky operators}

\author[B. Araneda]{Bernardo Araneda}
\address{Facultad de Matem\'atica, Astronom\'{\i}a, F\'{\i}sica y Computaci\'on\\
Universidad Nacional de C\'ordoba\\ 
Instituto de F\'{\i}sica Enrique Gaviola, CONICET\\
Ciudad Universitaria, (5000) C\'ordoba, Argentina} 
\email{baraneda@famaf.unc.edu.ar}
\date{December 12, 2019}

\thanks{The current version of this paper is based upon work supported by the Swedish Research 
Council under grant no. 2016-06596 while the author was in residence at Institut Mittag-Leffler in Djursholm, 
Sweden, during the fall 2019.}

\begin{abstract}

The Teukolsky equations are currently the leading approach for analysing stability of 
linear massless fields propagating in rotating black holes.
It has recently been shown that the geometry of these equations can be understood in 
terms of a connection constructed from the conformal and complex structure of Petrov type D spaces. 
Since the study of linear massless fields by a combination of conformal, complex and spinor methods 
is a distinctive feature of twistor theory, and since versions of the twistor equation 
have recently been shown to appear in the Teukolsky equations, this raises the question 
of whether there are deeper twistor structures underlying this geometry. 
In this work we show that all these geometric structures can be understood naturally by considering 
a {\em 2-dimensional} twistor manifold, whereas in twistor theory the standard (projective) twistor 
space is 3-dimensional.

\end{abstract}

\maketitle

\section{Introduction}

Twistor theory \cite{Pen67, PMC73} was originally conceived by Roger Penrose as a possible approach 
to quantum gravity, in which spacetime is no longer a fundamental entity but it is secondary to a more 
primitive structure. This structure is twistor space, which is (in its projective version) a three-dimensional 
complex manifold whose points correspond to `totally null 2-surfaces' in the spacetime.
The requirement that the twistor space so defined be {\em three}-dimensional forces the 
conformal curvature to be self-dual (SD) or anti-self-dual (ASD), which unfortunately is of little 
interest for the classical Lorentzian curved spacetimes of General Relativity.
In this work we study geometric constructions 
that a {\em two}- (rather than {\em three}-) dimensional moduli space of totally null 2-surfaces
induces on a 4-dimensional conformal structure, and their 
applications to the description of linear massless fields propagating on an algebraically special space.

\smallskip
Our main motivation comes from the apparently unrelated problem of black hole stability. 
The Teukolsky equations were found in \cite{Teu72, Teu73} 
and constitute currently the leading approach for analysing linear stability of massless 
fields propagating in a black hole spacetime. 
They are scalar, second order, partial differential equations involving only one component 
(in an appropriate frame) of the linear field under consideration.
The original derivation \cite{Teu73} is in terms of the Newman-Penrose (NP) formalism. 
One has to apply certain NP operators to the 
field equations written in NP form, and then make appropriate combinations of the 
resulting identities so as to obtain a differential equation for only one NP component of the field.
Even though there does not seem to be explicit geometric structures underlying this procedure,
in \cite{Bin02} it was found that, for the case of the Kerr spacetime, the Teukolsky equations 
have the form of a wave equation with potential in terms of a modified wave operator;
and in \cite{Aks11} this was generalized for all vacuum spacetimes of Petrov type D.
Generalized derivatives in physics appear naturally in gauge theories, where they indicate 
the presence of internal symmetries in the system and have a rich geometry associated to them; 
thus it is natural to ask whether the Teukolsky equations have such a geometric interpretation. 
Further interest in this question arises when taking into account the result found in \cite{Araneda2017} 
that certain spinor fields involved in the equations satisfy the twistor equation with 
respect to the Teukolsky derivative.
The problem of uncovering the underlying geometry was addressed in \cite{Araneda2018}, where, 
by using spinor methods, it was found that it can be understood from consideration of 
conformal and complex structures in the spacetime.
Now, since the combination of conformal, complex and spinor geometry in four dimensions is 
a natural arena for twistor theory,
the appearance in the same problem of (versions of) the twistor equation together 
with conformal and complex structures suggests that more profound aspects of twistor theory
could be involved in the problem.
This is further supported by the well-known result that twistor theory is especially 
powerful for studying massless free fields (although this is for the case of flat or 
(anti-)self-dual spacetimes).
Motivated by these facts, one of the main aims of this work is to demonstrate that 
deeper structures in twistor theory effectively underlie the geometry of the Teukolsky equations. 

\smallskip
Although the original developments in twistor theory were mainly concerned with the structure 
of General Relativity and its quantization, currently its main applications in physics 
are in the study of scattering amplitudes in particle physics and string theory 
(see the recent review \cite{Ati17}). 
Our results show that twistor methods can still be fruitfully applied 
to classical problems in General Relativity that are of current interest, and that they are very 
useful for the uncovering and understanding of geometric structures in these problems.

\subsection{Main results and overview}

The main result of this work is to establish a close relationship between 2-dimensional (2D) twistor 
manifolds and the Teukolsky (and related) equations.
This 
twistor manifold
is a 2D moduli space of totally null 2-surfaces,
and it has three crucial properties for us: 
$(i)$ it is associated to a {\em projective} spinor $[\xi^A]$ (we have an equivalence relation 
$\xi^A\sim\l\xi^A$),
$(ii)$ it is associated to a {\em conformal} structure $[g_{ab}]$ (we have an equivalence relation 
$g_{ab}\sim\Om^2 g_{ab}$), and
$(iii)$ it is a {\em complex} manifold (we have a complex structure $J^2=-1$).
These three properties are archetypal of a twistor space.

Section \ref{sec-TT} is a brief review of some basic aspects of twistor theory that are needed 
in the paper: the twistor equation, the definition of twistor space, and the Penrose transform 
for massless free fields.
Sections \ref{sec-2DTM} and \ref{sec-2TM} are devoted to our main results, where we study 
geometric constructions derived from the existence of 2D twistor spaces.
In section \ref{sec-csb} we show how a 2D twistor space $\T$ induces natural geometric structures 
in the spinor bundles
of a conformal manifold; in section \ref{sec-FB}, inspired by standard constructions in twistor theory,
we construct fibre bundles over $\T$ by using the previous geometric structures and their properties; 
and in section \ref{sec-te} we show how these constructions are related to the Teukolsky equations.
In particular, we show that line bundles over $\T$ give naturally solutions of these equations 
(for the case associated to massless free fields), in a manner that is reminiscent of the mechanisms 
involved in the Penrose transform.
Although gravitational perturbations are not included in this scheme for a number of reasons,
we make some comments regarding this case in section \ref{sec-gp}; in particular,
we show that metric reconstructions from Hertz potentials still admit a 2D twistor space, 
and we comment on possible consequences and applications of this result.
Finally we consider in section \ref{sec-2TM} the special case in which there are two independent
2D twistor spaces, which is naturally associated to Petrov type D spacetimes, and we reinterpret the 
previous constructions in terms of holomorphic structures.
We make some final remarks in section \ref{sec-FC}.

\subsection{Notation and conventions}

We work in 4-dimensional spacetimes $(\M,g_{ab})$ that admit a spinor structure and that 
are real-analytic, since we will often need to complexify the spacetime.
(See e.g. \cite[Section 6.9]{Penrose2} for the general rule when translating formulas from 
real to complex spacetimes.)
Our conventions follow those of Penrose and Rindler \cite{Penrose1, Penrose2}.
Indices $a,b,c,...$ are (abstract) 4-dimensional spacetime indices, while $A,B,...$ and 
$A',B',...$ are (abstract) 2-dimensional spinor indices. Boldface letters ${\bf A,B,...}$ 
etc. denote indices in a spin frame.
When considering complex spacetimes, the local Lorentz symmetry ${\rm SO}(1,3)$ 
is replaced by the complex rotations ${\rm SO}(4,\mbb{C})$. One has the isomorphism
\begin{equation}\label{so4c}
 {\rm SO}(4,\mbb{C})=(\SL_L\times\SL_R)/\mbb{Z}_2,
\end{equation}
where the subscripts $L,R$ mean `left' and `right' rotations, acting respectively on 
spinors with `unprimed' and `primed' indices.
The correspondence between vectors and spinors is via the soldering form, i.e.
$v^a\mapsto v^{AA'}=v^a\sigma_{a}{}^{AA'}$. This allows the identification 
$a\equiv AA'$, $b\equiv BB'$, etc., and in this work we will omit the soldering form 
$\sigma_{a}{}^{AA'}$.
Two complex-conjugate quantities $\Psi$ and $\bar\Psi$ that appear together in a 
real spacetime, become two independent quantities $\Psi$ and $\tilde\Psi$ in a complex 
spacetime; for example, the Weyl conformal spinor and its conjugate are independent 
entities $\Psi_{ABCD}$ and $\tilde{\Psi}_{A'B'C'D'}$ in the complex case.
Given a vector bundle $E$ over some manifold, the space of sections of $E$ 
will be denoted by $\G(E)$.

\section{Preliminaries on Twistor Theory}\label{sec-TT}

We review some basic aspects of the twistor equation in section \ref{sec-TE}, 
together with possible generalizations.
In section \ref{sec-TS} we give the definition of twistor space 
and its relation to spacetime by using the double fibration picture, both in the 
flat and in the curved spacetime case.
In section \ref{sec-PT} we recall the Penrose transform, that relates massless 
fields in the spacetime with sheaf cohomology classes over twistor space, and 
we give some explicit formulas for the fields in terms of cohomology elements.
(These constructions will be invoked in section \ref{sec-2DTM}.)
Except for section \ref{sec-TE}, we will work in {\em dual} twistor space (in the 
usual terminology of twistor theory).
The main references we follow in this section are \cite{Ada17, Penrose2, Ward, Huggett, Wit03}.

\subsection{The twistor equation}\label{sec-TE}

The twistor equation is\footnote{This subsection is related to the `usual' twistor space, 
i.e. not to its `dual' version, which is the one that we use in the rest of the paper.} 
\begin{equation}\label{teq}
 \c_{A'}{}^{(A}\w^{B)}=0,
\end{equation}
where $\w^A$ is a spinor field on a four-dimensional spacetime with spin structure and 
Levi-Civita connection $\c_{AA'}$.
In a flat spacetime, (\ref{teq}) can be thought of as a consequence of the `incidence relation', 
which is the (non-local) relation between points in spacetime and points in twistor space 
(see the next subsection).
In a curved spacetime, (\ref{teq}) imposes severe restrictions on the curvature: the 
integrability conditions are $\Psi_{ABCD}\w^D=0$, which for non-trivial $\w^A$
imply that the spacetime must be of Petrov type N or O. A possible generalization of (\ref{teq}) is 
\begin{equation}\label{ks}
 \c_{A'}{}^{(A}\w^{B...L)}=0,
\end{equation}
for some symmetric spinor field $\w^{A...K}$ with $n$ indices. Solutions to (\ref{ks}) 
are usually known as Killing spinors or twistor spinors.
A particularly relevant example of (\ref{ks}) corresponds to a 2-index Killing spinor, $\w^{AB}$, 
since it is well-known that all Einstein spacetimes of Petrov type D (in particular the Kerr solution) 
admit such object,  
which is associated to `hidden symmetries' in the spacetime and has found a lot of 
important applications both in past and recent years, see e.g. \cite{WP70, And09, And14, Fro17}.

Another possible generalization of (\ref{teq}) is to change the connection $\c_{AA'}$ to some other 
connection $D_{AA'}$,
\begin{equation}\label{wteq}
 D_{A'}{}^{(A}\w^{B)}=0,
\end{equation}
which can be regarded as a `charged' (or `weighted') twistor equation. As observed by Bailey 
\cite{Bailey, BaileyTN262, Bailey2}, this equation arises naturally for example in spacetimes 
that possess a shear-free null geodesic congruence; we will exploit this fact in section \ref{sec-2DTM}.
We also mentioned in the introduction that it
arises in the study of the Teukolsky equations: there exists a covariant derivative $D_a$ 
(the `Teukolsky connection') whose square $D^aD_a$ is the Teukolsky operator, and 
certain spinor fields involved in the equations satisfy (\ref{wteq}). (See the introduction 
in \cite{Araneda2018}.) This fact is actually one of the main motivations for the present work.

The approach to twistor theory by means of the twistor equation (\ref{teq}) (or its generalizations 
(\ref{ks}), (\ref{wteq})) emphasizes the use of {\em spinor fields on the spacetime} that satisfy 
differential equations. 
This point of view is perhaps not very convenient for the twistor treatment of curved spacetimes, 
since, as mentioned, the differential equations involved have integrability conditions that 
restrict the spacetime curvature.
Furthermore, in the original twistor programme, spacetime itself is a derived structure, 
that is secondary to the more primitive twistor space. 
This has profound implications in the nature of physical concepts; in particular, 
there is a non-local relation between points in spacetime and points in twistor space.
There are still (strong) restrictions on the curvature, but we find this construction of 
twistor space to be more suitable for the purposes of the present work.
Below we will briefly review the definition of twistor space as the moduli space of certain 
2-dimensional surfaces in the spacetime; 
this will proven to be more useful for the constructions studied in section \ref{sec-2DTM}.

\subsection{Twistor space}\label{sec-TS}

Let $\mbb{M}$ be (complexified) Minkowski spacetime.
Flat twistor space is $\mbb{T}\cong\mbb{C}^4$, and its coordinates are pairs 
of Weyl spinors of opposite quirality, $\mbb{T}\ni Z^{\a}=(\w^A,\pi_{A'})$. 
For our purposes it is more convenient to use instead {\em dual} twistor space, 
$\mbb{T}^{*}$, with coordinates $W_{\a}=(\l_A,\m^{A'})$. 
The relation between spacetime events $x^{AA'}\in\mbb{M}$ and points in 
$\mbb{T}^{*}$ is given by the so-called incidence relation: 
\begin{equation}\label{ir}
 \m^{A'}=-ix^{AA'}\l_{A}.
\end{equation}
(The twistor equation (\ref{teq}) is obtained by taking a spacetime derivative in the 
complex conjugate of (\ref{ir}).)
This equation remains true if we multiply $(\l_A,\m^{A'})$ by a non-zero complex 
number, so (\ref{ir}) actually defines a relation between spacetime and 
{\em projective} twistor space (we will generally omit the term `dual', and later 
also `projective'), $\mbb{PT}^{*}\cong\CP^3$, and one often works in 
this space instead of $\mbb{T}^*$.
If we fix $x^{AA'}$, then (\ref{ir}) defines a projective line $\CP^1$ in 
$\mbb{PT}^*$, whose topology is $S^2$. On the other hand, if we fix $(\l_A,\m^{A'})$, 
the set of $x^{AA'}$ that satisfy (\ref{ir}) turns out to be a 2-plane in $\mbb{M}$ that 
is totally null: every tangent to it has the form $\l^A\zeta^{A'}$ for fixed $\l^A$ 
and varying $\zeta^{A'}$. This 2-plane is called {\em $\b$-plane}.
Projective (dual) twistor space $\mbb{PT}^*$ is the space of $\b$-planes\footnote{If we
fix $\zeta^{A'}$ and vary $\l^A$ instead, the resulting 2-plane is an `$\a$-plane', and 
(projective) twistor space $\mbb{PT}$ is the space of $\a$-planes.}. 

The correspondence between twistor space and spacetime can be conveniently described 
via a double fibration. Let $\mbb{PS}_A$ be the projective spin bundle over $\mbb{M}$. 
The fibre over a point $x\in\mbb{M}$ is the projective space $\CP^1$. 
(Actually $\mbb{PS}_A$ is globally $\mbb{M}\times\CP^1$.)
The projection $\n$ over $\mbb{M}$ is simply $(x^a,\l_{A})\mapsto x^a$. 
$\mbb{PS}_A$ also projects to $\mbb{PT}^{*}$ by means of the incidence relation 
(\ref{ir}), i.e. via the map $\m$ given by $(x^a,\l_{A})\mapsto(\l_A,-ix^{AA'}\l_{A})$.
The double fibration is then
\begin{equation}\label{dfibr}
\begin{tikzcd}
 \mbox{} & \mbb{PS}_A \arrow[swap]{ld}{\m} \arrow{rd}{\n} & \mbox{} \\
 \mbb{PT}^{*} & \mbox{} & \mbb{M}
\end{tikzcd}
\end{equation}
This fibration represents the basic idea of twistor theory: Physics in the 
spacetime $\mbb{M}$ is translated into holomorphic data in twistor space $\mbb{PT}^{*}$.
One of the most prominent examples of this correspondence is the Penrose transform 
that we briefly review below.
Note that, similarly to the fact that the inverse image of a point $x\in\mbb{M}$ 
under $\n$ is the fibre $\n^{-1}(x)\cong\CP^1$, the inverse image of a point 
$(\l_A,\m^{A'})\in\mbb{PT}^{*}$ under $\m$ is the set of $x^{AA'}$ such that 
$\m^{A'}=-ix^{AA'}\l_{A}$, namely the whole $\b$-plane.

The (curved) twistor space associated to a {\em curved} spacetime is defined by generalizing 
the concept of $\b$-planes. (The resulting construction is known as the `Non-linear graviton' 
since the work of Penrose \cite{Pen76}.) A {\em $\b$-surface} in a complex spacetime $\M$ is a 
2-dimensional surface such that its tangent plane at each point is a $\b$-plane.
One can show (see the initial discussion in section \ref{sec-2DTM} below)
that the integrability conditions for the existence of a three-complex parameter 
family of $\b$-surfaces are $\Psi_{ABCD}\equiv0$, so the spacetime must be conformally 
half-flat (i.e. the conformal curvature must be SD).
The resulting 3-manifold is the (projective, dual) twistor space $\PTc^{*}$ of $\M$.
In the opposite direction, if the spacetime is SD, then one can see that it admits a 
complex 3-manifold of $\b$-surfaces, so the correspondence is one-to-one. 
Actually, the correspondence involves only the conformal structure of the spacetime, 
since the construction above is conformally invariant. 
By imposing additional conditions on $\M$, such as the vacuum Einstein equations, 
one obtains additional structures on $\PTc^{*}$. (We will not need these structures in 
this work; for details see e.g. \cite{Pen76}, \cite[Ch. 12]{Huggett} and \cite[Ch. 9]{Ward}.)
A double fibration picture like (\ref{dfibr}) relating $\PTc^{*}$ and $\M$ also applies, 
where the correspondence space is the projective spin bundle $\mbb{P}\mc{S}_{A}$.
Since a $\b$-plane is associated to a projective spinor $\l_A$, the set of $\b$-planes 
through a given point $x\in\M$ is parametrized by the projectivization of $\mbb{C}^2$, namely
$\CP^1$, thus, as in the flat case, a point in $\M$ corresponds to a projective line $\CP^1$ 
in $\PTc^{*}$. On the other hand, a point in $\PTc^{*}$ corresponds to a $\b$-surface in $\M$.

\subsection{The Penrose transform for massless fields}\label{sec-PT}

One of the most important results in twistor theory is the Penrose transform 
for massless fields: an isomorphism between solutions of the massless free field 
equations in the spacetime and certain sheaf cohomology groups over twistor space.
We recall that the massless free field equations of helicity $h$ are
\begin{subequations}
\begin{align}
 & \c^{AA'}\phi_{A'...F'}=0, \qquad \text{for } h>0 \label{RH} \\
 & \c^{AA'}\vp_{A...F}=0, \qquad \text{for } h<0 \label{LH} \\
 & \Box\vp=0, \qquad \text{for } h=0 \label{weq}
\end{align}
\end{subequations}
where the fields $\phi_{A'...F'}$ and $\vp_{A...F}$ are totally symmetric and 
have $2|h|$ indices each, and $\Box=\c^{AA'}\c_{AA'}$.
Solutions of (\ref{RH}) are called right-handed (RH) fields, and solutions of (\ref{LH})
are called left-handed (LH) fields.
In its original form the correspondence applies to Minkowski spacetime\footnote{There are 
important subtleties that we are omitting here, namely the fact that it is not actually the 
whole $\mbb{PT}^{*}$ which enters (\ref{Ptransform}) but the region with $\l_{A}\neq0$; 
we do not need to discuss this for the purposes of our presentation.}:
\begin{equation}\label{Ptransform}
 \left\lbrace\begin{matrix}\text{massless free fields} \\ 
 \text{of helicity $h$ in $\mbb{M}$} \end{matrix} \right\rbrace
 \cong \breve{H}^{1}(\mbb{PT}^{*},\mc{O}(2h-2)),
\end{equation}
where the right-hand side is a \u{C}ech cohomology group that we shortly discuss below.
The necessity of using cohomology can be understood by examining the representation 
of massless free fields as contour integrals of certain holomorphic functions over twistor space, 
since Penrose realized that the ``gauge'' freedom that one has in choosing these twistor functions 
is precisely that of a \u{C}ech representative of a cohomology class in $\mbb{PT}^{*}$.
The correspondence (\ref{Ptransform}) can be generalized to some extent to 
SD spacetimes (see \cite{Ward} and references therein for more details). 
More precisely, there is an isomorphism like (\ref{Ptransform}) for the case of 
negative helicity (i.e. LH fields), but for the case of positive helicity (RH fields) 
the analogous result involves {\em potentials} instead of the fields.
We will briefly review how to extract the spacetime field from a given cohomology element;
this will be useful in section \ref{sec-2DTM} for making some analogies between this 
procedure and the constructions thereof.
We work in a SD spacetime that satisfies the vacuum Einstein equations, i.e. such that 
$\Psi_{ABCD}=0$ and $\Phi_{ABA'B'}=0=\L$. This implies that we can use covariantly 
constant unprimed spinors, i.e. $\c_{AA'}\l_B=0$; below we will use this fact.
We found particularly useful the presentation in appendix A of \cite{Wit03}.

\medskip
One can describe the correspondence (\ref{Ptransform}) in terms of \u{C}ech or Dolbeault
cohomology; we will use the \u{C}ech approach here.
This is a cohomology theory based on a covering $\mc{U}=\{U_i\}$ of a topological space $X$. 
In order to introduce several concepts that we will be referring to below,
we now review in a rather elementary way some basic facts about \u{C}ech cohomology,
using notation that resembles closely the operations with differential forms and 
de Rham cohomology. (We follow mainly \cite{Wells, Ward, Penrose2}.)
A {\em sheaf} $\mc{S}$ (of abelian groups) over $X$ is essentially an assignment 
$U_i\to\mc{S}(U_i)$ of an abelian group $\mc{S}(U_i)$ (whose elements are called 
{\em sections of $\mc{S}$ over $U_i$}) to each open set $U_i$ in the covering $\mc{U}$, 
together with `restriction maps' $\mc{S}(U_i)\to\mc{S}(U_j)$ for $U_j\subset U_i$ and 
some additional conditions that we do not need to discuss here.
For example, if $E\to X$ is a vector bundle over $X$, the assignment $\mc{S}(U_i)=\G(U_i,E)$ 
(that is, the sections of $E$ over $U_i$) defines the so-called {\em sheaf of 
sections of the vector bundle $E$}.
Given $q+1$ sets $U_{i_0},...,U_{i_{q}}$ in $\mc{U}$ such that 
$U_{i_0}\cap...\cap U_{i_{q}}\neq\emptyset$, a {\em $q$-cochain} $f$ is a set of sections 
$\{f_{i_0...i_{q}}\}$ defined by 
$f_{i_{0}...i_{q}}:=f(U_{i_0},...,U_{i_{q}})\in\mc{S}(U_{i_0}\cap...\cap U_{i_{q}})$
that are totally antisymmetric, $f_{i_0...i_{q}}=f_{[i_0...i_{q}]}$.
The set of $q$-cochains is denoted by $C^q(\mc{U},\mc{S})$, and it is an abelian group 
(under pointwise addition).
Denoting the restriction of $f_{i_0...i_{q}}$ to $U_{i_0}\cap...\cap U_{i_{q}}\cap U_{i_{q+1}}$
by $f_{i_0...i_{q}}|_{i_{q+1}}$, the $q$-th coboundary operator, 
$\d^q:C^q(\mc{U},\mc{S})\to C^{q+1}(\mc{U},\mc{S})$, is defined by 
$(\d^q f)_{i_0...i_{q+1}}:=(q+1)f_{[i_0...i_q}|_{i_{q+1}]}$.
Since the composition $\d^{q+1}\circ\d^{q}$ vanishes, we have the complex 
$\cdots C^{q-1}(\mc{U},\mc{S})\to C^q(\mc{U},\mc{S}) \to C^{q+1}(\mc{U},\mc{S})\to \cdots$, 
and the cohomology of this complex gives the \u{C}ech cohomology groups.
More precisely, a {\em $q$-cocycle} is an element in the kernel of $\d^q$, that is 
$(\d^q f)_{i_0...i_{q+1}}=0$, and the set of $q$-cocycles is $Z^q(\mc{U},\mc{S}):=\ker\d^q$.
A {\em $q$-coboundary} is an element in the image of $\d^{q-1}$, that is 
$f_{i_{0}...i_{q}}=(\d^{q-1}h)_{i_0...i_{q}}$ for some $h\in C^{q-1}(\mc{U},\mc{S})$, 
and the set of $q$-coboundaries is $B^q(\mc{U},\mc{S}):={\rm im}\;\d^{q-1}$.
Then the {\em $q$-th \u{C}ech cohomology group}, with coefficients in the sheaf $\mc{S}$
and with respect to the covering $\mc{U}$, is defined as the quotient
\begin{equation*}
 \breve{H}^q(\mc{U},\mc{S}):=Z^q(\mc{U},\mc{S})/B^q(\mc{U},\mc{S}).
\end{equation*}
Under certain circumstances the \u{C}ech cohomology groups do not depend on the covering 
(these are called Leray covers); this will be the case below and so we can write 
$\breve{H}^{q}(X,\mc{S})$. The topological space in our 
context is projective (dual) twistor space, but in practice, using the double fibration 
(\ref{dfibr}) we will only need cohomology over a projective line, so $X=\CP^1$. 
(This space can be covered by two open sets: $U_0=\{\l_{A}|\l_{0}\neq0\}$ and 
$U_1=\{\l_{A}|\l_{1}\neq0\}$.) 
Over $\CP^1$ one defines the complex line bundles $\mc{O}(k)$, $k\in\mbb{Z}$, whose 
sections are complex-valued functions homogeneous of degree $k$ in the 
homogeneous coordinates of $\CP^1$, that is $f(z\l_A)=z^k f(\l_A)$.
The sheaf $\mc{S}$ will be the sheaf of sections of $\mc{O}(k)$,
which is also denoted by $\mc{O}(k)$.

We will only need the zeroth and first cohomology groups. By construction, the 
0-th cohomology group coincides with the space of global sections of the sheaf. 
In our case one can show that (see e.g. Example 2.13 in \cite[Chapter I]{Wells})
\begin{equation}\label{H0}
 \breve{H}^0(\CP^1,\mc{O}(k))=
 \begin{cases} 
  0, \qquad k<0 \\
  \mbb{C}, \qquad k=0 \\
  \text{homogeneous polynomials of degree $k$ in }\mbb{C}^2, \qquad k>0
 \end{cases}
\end{equation}
For the first cohomology group one has
\begin{equation}\label{H1}
 \breve{H}^1(\CP^1,\mc{O}(k))=
 \begin{cases} 
  \mbb{C}^{-k-1}, \qquad k<-1 \\
  0, \qquad k\geq-1 
 \end{cases}
\end{equation}

Suppose $\PTc^{*}$ is covered by open sets $\mc{V}_i$.
A cohomology class in $\breve{H}^1(\PTc^{*},\mc{O}(2h-2))$ is represented by 
a 1-cocycle $f_{ij}$ (modulo coboundaries).
It is convenient to think of $f_{ij}$ as a function on the spin bundle by means of its 
pull-back by $\m$, using the (curved version of the) double fibration (\ref{dfibr}) 
(see e.g. \cite[Section 9.1]{Ward}). 
More precisely, let $V_i=\m^{-1}(\mc{V}_i)$, which is an open set on the spin bundle. 
We think of $f_{ij}$ as a function on $V_i\cap V_j$, $f_{ij}(x,\l)$, which is homogeneous 
of degree $2h-2$ in $\l_A$, and constant on $\b$-surfaces:
\begin{equation}\label{twf0}
 \c_{X}f_{ij}(x,\l)=0
\end{equation}
for all $X$ tangent to the $\b$-surface associated to $\l^A$. 
Since these tangents are of the form $X^a=\l^A\zeta^{A'}$ for arbitrary $\zeta^{A'}$, this is 
equivalent to
\begin{equation}\label{twf}
 \l^{A}\c_{AA'}f_{ij}(x,\l)=0.
\end{equation}

\begin{remark}\label{rem-we}
Taking an additional derivative $\c_{B}{}^{A'}$ in (\ref{twf}), we see that $f_{ij}$
solves the wave equation
\begin{equation}\label{we}
\Box f_{ij}(x,\l)=0.
\end{equation}
We will invoke this fact later on when studying 2-dimensional twistor manifolds; 
in particular, we will see that the Teukolsky equations are the natural generalization of 
(\ref{we}) in this context (see remark \ref{rem-teq} below).
\end{remark}

Now, for fixed $x\in\M$, $f_{ij}$ can be thought of as a 1-cocycle in 
$\breve{H}^1(\CP^1,\mc{O}(2h-2))$. Consider first the case of positive helicity.
From the $k\geq-1$ case in (\ref{H1}) we know that $\breve{H}^1(\CP^1,\mc{O}(2h-2))=0$ 
for $h\geq1/2$. This implies that $f_{ij}$ is a coboundary, i.e. it splits as 
$f_{ij}(x,\l)=h_i(x,\l)-h_j(x,\l)$, where $h_i$ is holomorphic on $V_i$ and $h_j$ 
is holomorphic on $V_j$. Using (\ref{twf}), we deduce that 
\begin{equation}\label{gf}
 \l^{A}\c_{AA'}h_i(x,\l)=\l^{A}\c_{AA'}h_j(x,\l).
\end{equation}
This equation defines a (spinor-valued) global function in $\CP^1$, homogeneous 
of degree $2h-1$ (with $h\geq1/2$), i.e. an element of $\breve{H}^0(\CP^1,\mc{O}(2h-1))$. 
From this we can extract the RH fields as follows.

For $h=1/2$, (\ref{gf}) defines an element of $\breve{H}^0(\CP^1,\mc{O}(0))$. 
From the $k=0$ case in (\ref{H0}), we deduce that (\ref{gf}) must be constant 
as a function of $\l_{A}$, so we get a field on the spacetime:
\begin{equation}\label{RHD}
 \phi_{A'}:=\l^{A}\c_{AA'}h_i(x,\l)
\end{equation}
Now, we have $\c^{AA'}\phi_{A'}=\frac{1}{2}\l^{A}\Box h_i(x,\l)$. 
Equation (\ref{we}) implies $\Box h_i(x,\l)=\Box h_j(x,\l)$, 
but this last equation defines a global function in $\CP^1$ homogeneous of degree 
$-1$, i.e. an element of $\breve{H}^0(\CP^1,\mc{O}(-1))$, 
so from the $k<0$ case in (\ref{H0}) we see that it must be zero: $\Box h_i(x,\l)=0$. 
Therefore we get a massless RH Dirac field on $\M$, $\c^{AA'}\phi_{A'}=0$. 

For $h=1$, i.e. for RH Maxwell fields, the procedure is similar to the Dirac 
case except that, as mentioned, we must now use potentials. Equation (\ref{gf}) 
defines an element of $\breve{H}^0(\CP^1,\mc{O}(1))$. 
From the $k>0$ case in (\ref{H0}) we deduce that its dependence in 
$\l_{A}$ must be polynomial, so we get
\begin{equation}\label{RHvp}
 \l^{A}\c_{AA'}h_i(x,\l)\equiv \l^{A}A_{AA'}(x),
\end{equation}
introducing in this way a covector field $A_{a}$ on the spacetime. Operating on 
(\ref{RHvp}) with $\l^{B}\c_{B}{}^{A'}$, on the LHS we get 
$\l^{B}\l^{A}\c_{B}{}^{A'}\c_{A'A}h_i=0$ (since $\c_{(B}{}^{A'}\c_{A)A'}h_i=0$), 
thus on the RHS we have $\c_{A'(A}A_{B)}{}^{A'}=0$. 
Now, the 2-form $F_{ab}:=2\c_{[a}A_{b]}$ satisfies $\c_{[b}F_{cd]}=0$, so multiplying by 
$\e^{abcd}$ we get $\c_{b}{}^{*}F^{ab}=0$. But the spinor decomposition of $F_{ab}$
is $F_{ab}=\psi_{AB}\e_{A'B'}+\phi_{A'B'}\e_{AB}$ with $\psi_{AB}=\c_{A'(A}A_{B)}{}^{A'}$
and $\phi_{A'B'}=\c_{A(A'}A_{B')}{}^{A}$, so since $\psi_{AB}=0$, $F_{ab}$ is SD:
${}^{*}F_{ab}=iF_{ab}$, therefore $\c_{b}F^{ab}=0$ or, equivalently, $\c^{AA'}\phi_{A'B'}=0$,
i.e. we get a RH Maxwell field on $\M$.

\smallskip
For $h>1$ the existence of RH fields is constrained by the well-known Buchdahl conditions 
involving the SD curvature.
For $h<0$, say $n=-2h>0$, there are no constraints since by assumption the spacetime is SD, 
namely $\Psi_{ABCD}\equiv0$.
In this case, to extract the LH fields, we consider again an element $f_{ij}\in\breve{H}^1(\CP^1,\mc{O}(-n-2))$ 
as a function on the spin bundle that is homogeneous in $\l_A$ of degree $-n-2$ and 
satisfies (\ref{twf}). Now, the field 
\begin{equation}\label{TLH}
 \Phi_{ijA...L}(x,\l):=\l_A...\l_L f_{ij}(x,\l),
\end{equation}
with $n+1$ factors of $\l_A$, satisfies $\c^{AA'}\Phi_{ijA...L}(x,\l)=0$ by virtue of (\ref{twf}).
Furthermore it is homogeneous of degree $-1$ in $\l_A$, so it can be regarded as a (spinor-valued)
element of $\breve{H}^1(\CP^1,\mc{O}(-1))$.
By the $k=-1$ case in (\ref{H1}), this group is trivial so (\ref{TLH}) must split as 
$\Phi_{ijA...L}(x,\l)=h_{iA...L}(x,\l)-h_{jA...L}(x,\l)$, with $h_{iA...L}(x,\l)$ holomorphic on $V_i$ 
and $h_{jA...L}(x,\l)$ holomorphic on $V_j$.
Taking a derivative, we get $\c^{AA'}h_{iA...L}(x,\l)=\c^{AA'}h_{jA...L}(x,\l)$, but this defines a 
global function in $\CP^1$ that is homogeneous of degree $-1$, so it must be zero.
Similarly, contracting (\ref{TLH}) with $\l^L$, we get $h_{iA...KL}(x,\l)\l^{L}=h_{jA...KL}(x,\l)\l^{L}$, 
and this is a global function in $\CP^1$, homogeneous of degree $0$, so it does not depend on $\l_A$, 
therefore $\vp_{A...K}(x):=h_{iA...KL}(x,\l)\l^{L}$ is a field on the spacetime and 
satisfies the LH massless free field equations (\ref{LH}).
The procedure above is the cohomological version of the well-known contour integral formula 
of Penrose.

\smallskip
The examples considered above are just some well-known instances (the ones that we will invoke 
later on in this paper) of the powerful methods of twistor theory, that involve {\em linear} field equations.
Twistor methods have also been extremely useful in the study of {\em non-linear} differential equations.
For example, they have led to a one-to-one correspondence between solutions 
of the SD or ASD Yang-Mills equations and holomorphic vector bundles over twistor space 
that are trivial on each projective line; this is known as the Ward transform. 
The Non-linear graviton (referred to above) is another example, which establishes a one-to-one 
correspondence between 4-dimensional SD manifolds satisfying the vacuum Einstein equations, 
and twistor spaces with some additional structures.
We will not need these non-linear constructions in the present work.

\section{Two-dimensional twistor spaces}\label{sec-2DTM}

We will now study the geometry associated to the existence of a complex {\em 2-dimensional} 
(rather than {\em 3-dimensional}) moduli space of totally null 2-surfaces. 
Our main goal is to show that this twistor structure, which is present in, for example, 
all conformally Einstein, algebraically special spaces, gives a natural geometric structure 
to several constructions associated to the description of massless fields propagating in 
curved spacetimes, and is in particular closely related to the geometry of the Teukolsky equations
and black hole perturbation theory.

\smallskip
We recall that a totally null 2-surface $\Sigma$ on a complex spacetime $(\M,g_{ab})$ 
(already introduced in section \ref{sec-TS}) is 
a complex 2-surface such that, for any two vectors $u^a$, $v^a$ 
tangent to $\Sigma$ at a point $p\in\Sigma$, it holds $g_{ab}u^a v^b=0$. 
Note that this condition is conformally invariant (i.e. it remains true if we make the 
transformation $g_{ab} \to \Om^2 g_{ab}$),
thus a totally null 2-surface is actually associated to the conformal structure of the spacetime, 
so henceforth we assume that we are working on a conformal manifold $(\M,[g])$.
The tangent vectors to $\Sigma$ are of the form $\xi^A\m^{A'}$, where either $\xi^A$ is fixed 
and $\m^{A'}$ varies (in which case $\Sigma$ is called $\b$-surface), or $\xi^A$ varies and 
$\m^{A'}$ is fixed (in which case $\Sigma$ is an $\a$-surface). 
We will focus here on $\b$-surfaces. By Frobenius theorem,
the condition for $\Sigma$ to be indeed a 2-surface is equivalent to the statement that, 
given any two vectors $u^a=\xi^A\m^{A'}$, $v^a=\xi^A\n^{A'}$, 
tangent to $\Sigma$ at $p\in\Sigma$, their Lie bracket should be a linear combination of them, 
namely $[u,v]^a=a u^a+b v^a$ for some scalar fields $a,b$. 
In other words, we must have $[u,v]^a=\xi^A\zeta^{A'}$ for some $\zeta^{A'}$. Replacing 
the expressions for $u^a$ and $v^a$, in general one finds
\begin{equation*}
 [u,v]^a=(\m^{C'}\n_{C'})\xi^B\c_{B}{}^{A'}\xi^A+\xi^A\k^{A'},
\end{equation*}
with $\k^{A'}=\xi^B(\m^{B'}\c_{BB'}\n^{A'}-\n^{B'}\c_{BB'}\m^{A'})$, and $\c_{AA'}$ 
the Levi-Civita connection of an arbitrary metric in the conformal class.
Thus the condition $[u,v]^a=\xi^A\zeta^{A'}$ is satisfied for {\em any} $u^a$ and $v^a$
tangent to $\Sigma$ if and only if 
$\xi^B\c_{B}{}^{A'}\xi^A=\xi^A\pi^{A'}$ for some $\pi^{A'}$, or equivalently, 
if and only if $\xi^A$ satisfies
\begin{equation}\label{SFR}
 \xi^A\xi^B\c_{AA'}\xi_B=0.
\end{equation}
This is exactly the condition for the null congruence associated to $\xi^A$ to 
be geodesic and shear-free\footnote{Note that, consistently, equation (\ref{SFR}) is 
conformally invariant if $\xi^A$ has well-defined conformal weight.} (SFR from now on). 
We thus arrive at the following result of Penrose and Rindler \cite{Penrose2} 
(we rephrase it according to our context):

\begin{prop}[Proposition (7.3.18) in \cite{Penrose2}]\label{prop-PR}
 A (complexified) conformal structure admits a 2-complex dimensional moduli space of 
 totally null 2-surfaces if and only if it admits a shear-free null geodesic congruence.
\end{prop}

\noindent
Considering a spin frame $(\xi^A,\eta^A)$ (with $\xi_A\eta^A=1$) and using standard notation 
for GHP spin coefficients (see e.g. \cite[Eq. (4.5.21)]{Penrose1}), in an arbitrary spacetime we have
\begin{equation}
 \xi^A\c_{A}{}^{B'}\xi^{B}=\xi^B\pi^{B'}+\eta^B(\k\iota^{B'}-\sigma o^{B'})
\end{equation}
where $\pi^{B'}=\b o^{B'}-\e\iota^{B'}$ and $(o^{A'},\iota^{A'})$ is a primed spin frame.
We thus see that $\xi^{A}$ is an SFR if and only if the following conditions hold:
\begin{equation}\label{SFR2}
 \k=0=\sigma.
\end{equation}

The integrability conditions for (\ref{SFR}) are $\Psi_{ABCD}\xi^A\xi^B\xi^C\xi^D=0$. 
If we require this to hold for {\em any} spinor $\xi^A$ at any point of $\M$, then 
we must have $\Psi_{ABCD}\equiv 0$, i.e. the conformal structure must be SD. 
The resulting {\em three}-complex parameter family of $\b$-surfaces
is the (curved, projective, dual) twistor space $\PTc^*$ of the conformal structure $(\M,[g])$, 
that we introduced at the end of section \ref{sec-TS}.
Proposition \ref{prop-PR} tells us that the existence of a {\em two}-complex parameter family 
of $\b$-surfaces is equivalent to the existence of an SFR, which is a much weaker condition.
This is a 2D twistor space and we will denote it by $\T$.

If we assume that the condition $\Psi_{ABCD}\xi^A\xi^B\xi^C\xi^D=0$ is valid for 
{\em a particular} spinor field $\xi^A$, this means
that $\xi^A$ must be a principal null direction (PND) of the ASD Weyl spinor.
Eventually we will also require the stronger condition $\Psi_{ABCD}\xi^B\xi^C\xi^D=0$, namely, 
that $\xi^A$ be a two-fold PND of $\Psi_{ABCD}$.
By the Goldberg-Sachs theorem, this is automatically satisfied in all conformal structures 
with an SFR that admit an Einstein metric.

\subsection{Structures on the conformal spinor bundles}\label{sec-csb}

From proposition \ref{prop-PR}, the existence of a 2D twistor space $\T$ singles out a spinor 
field $\xi^A$ in the (conformal) spacetime. We will show that a choice of a preferred spinor defines 
natural connections on spinor and tensor bundles in the conformal structure\footnote{Our main 
reference for concepts and definitions regarding conformal geometry is \cite{VK16}.}. 
This is independently of $\xi^A$ being or not an SFR; the SFR condition becomes relevant 
when studying additional properties of the associated connection such as its curvature.

\smallskip
As preliminaries, consider a complexified spacetime and denote by $(\M,[g])$ its conformal structure. 
The set of all frames $\{e_{\bf a}\}$ 
(${\bf a}=0,...,3$) such that $g(e_{\bf a}, e_{\bf b})=\Om^2 \eta_{\bf ab}$, 
with $g\in[g]$, $\Om\in\mbb{R}^{+}$ and $\eta_{\bf ab}={\rm diag}(1,-1,-1,-1)$, 
gives a principal fibre bundle with structure group ${\rm SO}(4,\mbb{C})\times\mbb{R}^{+}$.
The associated spin structure\footnote{Recall that the spin structure of a 
conformal manifold is a well-defined concept, see e.g. \cite[Section 5.6]{Penrose1}
and also Note 8 to Chapter 9 in \cite{Mason96}.} 
is denoted by $P_{\rm Spin}$,
and its structure group is $G=\SL_{L}\times\SL_{R}\times\mbb{R}^{+}$, where the two
factors of $\SL$ account for `left' and `right' rotations (recall (\ref{so4c})), and
the group $\mbb{R}^{+}$ corresponds to conformal transformations of the metric. 
If $\ve^A_{\bf A}=\{\ve^A_0,\ve^A_1\}$ and $\ve^{A'}_{\bf A'}=\{\ve^{A'}_{0'},\ve^{A'}_{1'}\}$
are unprimed and primed spin frames respectively, we choose their conformal 
behavior as $\wh\ve^A_0=\Om^{w_0}\ve^A_0$, $\wh\ve^A_1=\Om^{w_1}\ve^A_1$, 
and $\wh\ve^{A'}_{0'}=\Om^{w_0}\ve^{A'}_{0'}$, $\wh\ve^{A'}_{1'}=\Om^{w_1}\ve^{A'}_{1'}$, 
with $w_0+w_1+1=0$ (so that $\wh{\e}^{AB}=\Om^{-1}\e^{AB}$ and $\wh{\e}^{A'B'}=\Om^{-1}\e^{A'B'}$).
Considering the representation of $\mbb{R}^{+}$ in $\mbb{C}^2$ given by 
$\Om\mapsto{\rm diag}(\Om^{w_0}, \Om^{w_1})$, this means that the group $G$ acts on a 
spin frame $\ve_{\bf A}^A$ as $\ve_{\bf A}^A\mapsto C_{\bf A}{}^{\bf B}\ve_{\bf A}^B$, where 
$C_{\bf A}{}^{\bf B}$ is the product between a matrix $S$ of $\SL_L$ and ${\rm diag}(\Om^{w_0}, \Om^{w_1})$;
similarly for the primed spin frame $\ve_{\bf A'}^{A'}$.

Consider the vector space $V^{k,k'}_{l,l'}=(\mbb{C}^2)^{\otimes k}\otimes(\bar{\mbb{C}}^2)^{\otimes k'}
\otimes(\mbb{C}^{2*})^{\otimes l}\otimes(\bar{\mbb{C}}^{2*})^{\otimes l'}$, and the representation
$\varrho:\SL_{L}\times\SL_{R}\times\mbb{R}^{+}\to{\rm GL}(V^{k,k'}_{l,l'})$ defined by
\begin{equation}\label{rep}
 (\varrho(S,\tilde{S},\Om)\Psi)^{\bf A...A'...}_{\bf B...B'...}=
 C_{\bf P}{}^{\bf A}...\tilde{C}_{\bf P'}{}^{\bf A'}...(C^{-1})_{\bf B}{}^{\bf Q}...(\tilde{C}^{-1})_{\bf B'}{}^{\bf Q'}
 ...\Psi^{\bf P...P'...}_{\bf Q...Q'...}
\end{equation}
where $C_{\bf P}{}^{\bf A}$ is as before and $\Psi\in V^{k,k'}_{l,l'}$.
Then one can construct the associated vector bundles
\begin{equation}\label{sb}
 \mc{S}^{k,k'}_{l,l'}:=P_{\rm Spin}\times_{\varrho}V^{k,k'}_{l,l'},
\end{equation}
the sections of which are spinor fields on $\M$.
For example, the cases $\mc{S}^{A}\equiv\mc{S}^{1,0}_{0,0}$ and $\mc{S}^{A'}\equiv\mc{S}^{0,1}_{0,0}$ 
correspond to the unprimed and primed spin bundles respectively, and the case $\mc{S}^{AA'}\equiv\mc{S}^{1,1}_{0,0}$ 
can be identified with the tangent bundle of the manifold $\M$.
Using the abstract index notation, a section $\Psi\in\G(\mc{S}^{k,k'}_{l,l'})$ is
\begin{equation}
 \Psi^{A...A'...}_{B...B'...}= \Psi^{\bf A...A'...}_{\bf B...B'...}\ve^{A}_{\bf A}...
 \ve^{A'}_{\bf A'}...\ve^{\bf B}_{B}...\ve^{\bf B'}_{B'}.
\end{equation}
Considering also the standard construction of conformally weighted line bundles $\mc{E}[w]$ (whose sections 
are conformal scalar densities with weight $w$), and taking the tensor product 
$\mc{S}^{k,k'}_{l,l'}\otimes\mc{E}[w]=:\mc{S}^{k,k'}_{l,l'}[w]$,
the sheaf of sections $\G(\mc{S}^{k,k'}_{l,l'}[w])$ gives conformally weighted spinor fields.

\subsubsection{Conformal connections}

Let $P\to\M$ be a principal bundle over $\M$, with structure group $G$. A connection on $P$ 
is a decomposition of the tangent bundle of $P$ as a direct sum $TP=TV\oplus TH$ of `vertical' 
and `horizontal' bundles. 
The vertical bundle $TV$ is naturally defined and is isomorphic to the Lie algebra $\mf{g}$
of $G$. The horizontal bundle can be defined by using a connection 1-form, 
which is a 1-form $\w\in T^{*}P\otimes\mf{g}$ such that $TH=\ker \w$. 
Given an open neighbourhood $U\subset\M$, a {\em local} connection form is a
$\mf{g}$-valued 1-form $A$ over $U$. If $\sigma:U\to P$ is a local section over $U$, 
then there exists a connection form in $P$ such that $A=\sigma_{*}\w$; in what 
follows we will focus on local connection forms.
On the other hand, a connection on a vector bundle $E$ over $\M$ (which we also refer to as 
a covariant derivative) is essentially a linear map $\G(E)\to\G(E\otimes T^*\M)$ that satisfies 
the Leibniz rule.
Given a representation $\varrho:G\to{\rm GL}(V)$ of $G$ on a vector space $V$, we can construct 
associated vector bundles as $E=P\times_{\varrho}V$.
A natural way to get a connection on $E$ is to use the connection 1-form of $P$ or, 
rather, the local connection $A$. More precisely, if $\varrho'$ is the representation of 
the Lie algebra $\mf{g}$ associated to $\varrho$, then one can show (see e.g. \cite[Chapter 10]{Nak03})
that the connection induced on $E$ is 
\begin{equation}\label{IC}
\p_a+\varrho'(A_a).
\end{equation}

For a fixed spacetime, a trivial example of this construction is to take $P={\rm SO}\M$ 
(the orthonormal frame bundle) and the natural representation of ${\rm SO}(4,\mbb{C})$ in $\mbb{C}^4$,
then we can view the tangent bundle as $T\M\cong{\rm SO}\M\times_{{\rm SO}(4,\mbb{C})}\mbb{C}^4$.
The local connection 1-form in ${\rm SO}\M$ is the spin connection $\varpi_a$, thus the 
Levi-Civita connection $\c_a$ on $T\M$ can be viewed as induced from $\varpi_a$ in the 
manner (\ref{IC}), and the construction generalizes easily to tensor bundles over $\M$.
Of course, for tensor fields this is just a sophisticated way of describing their covariant 
derivative, but, as is well-known, the construction is essential when dealing with 
spinors (or more generally with gauge theories), since the only sensible way of defining
spinor fields is via associated bundles such as (\ref{sb}), and similarly for fields 
with internal degrees of freedom. 
This will be the approach that we use here for inducing natural connections on bundles over $\M$
from the 2D twistor space $\T$.

Now, if instead of a fixed spacetime we have the conformal structure $(\M,[g])$,
then a sensible analog of the Levi-Civita connection is a Weyl 
connection, which is a pair $(\slashed{\c}{}_{a},\mf{f}_a)$ consisting of a torsion-free connection 
$\slashed{\c}{}_{a}$ and a 1-form $\mf{f}_a$ such that for any representative $g_{ab}$ 
of the conformal class, it holds $\slashed{\c}{}_a g_{bc}=-2\mf{f}_a g_{bc}$, where $\mf{f}_a$ 
transforms under change of conformal representative (i.e. $g_{ab}\to\Om^2g_{ab}$) as 
$\mf{f}_a\mapsto \mf{f}_a-\U_a$, with $\U_a=\Om^{-1}\c_a\Om$. 
For a spinor field $\Psi\in\G(\mc{S}^{k,k'}_{l,l'})$, the relation 
between $\slashed{\c}{}_{a}$ and a Levi-Civita connection $\c_a$ is given by
\begin{align}
\nonumber \slashed{\c}{}_{a}\Psi^{B...B'...}_{C...C'...} =\;& \c_{a}\Psi^{B...B'...}_{C...C'...}
 +\e_{A}{}^{B}\mf{f}_{A'E}\Psi^{E...B'...}_{C...C'...}+...+\e_{A'}{}^{B'}\mf{f}_{AE'}\Psi^{B...E'...}_{C...C'...}+...\\
 & -\mf{f}_{A'C}\Psi^{B...B'...}_{A...C'...}-...-\mf{f}_{AC'}\Psi^{B...B'...}_{C...A'...}-.... \label{weylc}
\end{align}
More generally, for spinor fields with non-trivial conformal weight, this does not give a connection on 
$\mc{S}^{k,k'}_{l,l'}[w]$ since (\ref{weylc}) does not transform covariantly under 
conformal transformations. Instead, the appropriate connection is now
\begin{equation}
 \slashed{\c}{}_{a}\Psi^{B...B'...}_{C...C'...}+w\mf{f}_a\Psi^{B...B'...}_{C...C'...}.
\end{equation}

The problem now is that, unlike the Levi-Civita connection, Weyl connections are in 
principle not unique. There are some situations however where a preferred Weyl connection 
is singled out by particular properties of the system under consideration. 
This is for example the case when studying conformal geodesics, see e.g. \cite[Section 5.5]{VK16}.
Another example occurs in a conformal almost-Hermitian manifold, namely in a conformal 
structure that is also equipped with a compatible almost-complex structure $J$, which is a 
tensor field $J_{a}{}^{b}$ such that $J_{a}{}^{c}J_{c}{}^{b}=-\d_{a}{}^{b}$ and 
$J_{a}{}^{c}J_{b}{}^{d}g_{cd}=g_{ab}$ for any $g_{ab}$ in the conformal class, 
see \cite{Bailey3, Gover2013}.
In this situation, there exists a {\em unique} Weyl connection compatible with $J$, 
where `compatible' means that such Weyl connection, here denoted $(\slashed{\c}{}_a,\mf{f}_a)$, 
is determined uniquely by requiring that $\slashed{\c}{}_b J_{a}{}^{b}=0$ 
(see e.g. \cite[Section 4]{Gover2013}). 
In terms of the Levi-Civita connection $\c_a$ of a conformal representative $g_{ab}$,
$\mf{f}_a$ is given by $\mf{f}_a=-\frac{1}{2}J_{b}{}^{c}\c_{c}J_{a}{}^{b}$.
($\mf{f}_a$ is sometimes called the {\em Lee form}.)
In the present work we are dealing with complexified spacetimes, which, by definition, 
already have a complex structure;
but we will see below that the 2D twistor space $\T$ induces a {\em canonical} 
almost-complex structure (and this is also true for the {\em real} Lorentzian spacetime 
we started from). 
Consequently, we will obtain from $\T$ a canonical Weyl connection.

\subsubsection{Induced canonical complex structure}

From proposition \ref{prop-PR}, the 2D twistor space $\T$ defines a preferred 
spinor field $\xi^A$ in the spin bundle $\mc{S}^A$. 
We choose $\xi^A$ as an element of a spin frame, $\xi^A\equiv\ve^A_{0}$. 
Let $\eta^A$ be any other spinor field such that $\e_{AB}\xi^A\eta^B=1$ for any choice 
of conformal spin metric $\e_{AB}$; thus $(\xi^A,\eta^A)$ is a spin frame, the conformal 
weights of $\xi^A$ and $\eta^A$ being, respectively, $w_0$ and $w_1$, with $w_0+w_1+1=0$.
Since $\T$ determines $\xi^A$ only up to multiples, we have the freedom 
$\xi^A\to\l\xi^A$, with $\l$ a complex number different from zero. In turn, for $\eta^A$
we have the freedom $\eta^A\to \l^{-1}\eta^A+b\xi^A$, where $b$ is any complex number.
This means that the gauge group $\SL_{L}$ is reduced to $\mbb{C}^{\times}\times\mbb{C}^{+}$, 
where $\mbb{C}^{\times/+}$ is the multiplicative/additive group of complex 
numbers\footnote{Note that the spin group $\SL$ can be decomposed as 
$\SL\cong\mbb{C}^{\times}\times\mbb{C}^{+}\times\mbb{C}^{+}$, where $\mbb{C}^{\times}$ is 
the `GHP part' and the two factors of $\mbb{C}^{+}$ correspond to null rotations around the spinors 
of the frame.}.
Now, for any $x\in\M$, consider the linear operator $J:T_x\M\to T_x\M$ given by
\begin{equation}\label{J}
 v^a \mapsto J_{a}{}^{b}v^a:=i(\xi_{A}\eta^B+\eta_A\xi^B)\e_{A'}{}^{B'}v^{AA'}.
\end{equation}
Then it is straightforward to show that $J_{a}{}^{c}J_{c}{}^{b}=-\d_{a}^{b}$ and 
$J_{a}{}^{c}J_{b}{}^{d}g_{cd}=g_{ab}$ (with $g_{ab}=\e_{AB}\e_{A'B'}$), 
so (\ref{J}) equips $T_x\M$ with an almost-complex structure compatible 
with the conformal metric\footnote{Note that (\ref{J}) is a {\em complex} map, whereas 
the usual notion of an almost-complex structure requires it to be real. However, as shown 
in Theorem VIII.3 in \cite{Fla76}, a Lorentzian manifold (which is ultimately the most interesting 
case for our purposes) cannot admit a (real) almost-Hermitian 
structure, so we are forced to consider this complex-valued almost-Hermitian structure 
(in \cite{Fla76} this is referred to as a `modified' Hermitian structure). We will give an interpretation 
of (\ref{J}) in section \ref{sec-hvb} below.}.
(We note that a complex structure formally analogous to (\ref{J}) is used in \cite[Section 9.1]{Ward}
for the construction of the twistor space of a Riemannian ---i.e. positive definite--- 4-manifold,
where the spinor $\eta^A$ is obtained via an antiholomorphic involution applied to $\xi^A$; 
see equation (9.1.20) in that reference.) 

Of course, the map (\ref{J}) depends on a choice of $\eta^A$, with $\xi_A\eta^A=1$. 
Suppose an arbitrary choice of such an $\eta^A$ is made.
Since the null direction associated to $\eta^A$ is not fixed by the geometry, in principle we could 
change $\eta^A$ to $\eta^A+b\xi^A$. But the map (\ref{J}) then transforms to
$J_{a}{}^{b}+2i b\xi_A\xi^B\e_{A'}{}^{B'}$, which, if $b\neq0$,
depends explicitly on the choice of a representative from the projective class $[\xi^A]$.
Therefore, if we want the complex structure (\ref{J}) to depend only on the projective class 
of $\xi^A$, then we have to set $b=0$, which means that the gauge group 
$\mbb{C}^{\times}\times\mbb{C}^{+}$ is further reduced to $\mbb{C}^{\times}$.
(In other words, once we have arbitrarily chosen an $\eta^A$ with $\xi_A\eta^A=1$, 
the requirement that (\ref{J}) should depend only on $[\xi^A]$
does not allow us to make the transformation $\eta^A\to \eta^A+b\xi^A$.)

\begin{remark}
At this point, the fixing of the almost-complex structure is required in order to get a canonical Weyl connection. 
But (\ref{J}) and the structures derived from it are actually interesting on its own; 
we will see more about this in section \ref{sec-hvb} below.
\end{remark}

Now, fixing $J$ has two effects: on the one hand, it reduces the $\SL_L$ part of the gauge group 
to the GHP group $\mbb{C}^{\times}$, and on the other hand, determines a canonical 
Weyl connection $(\slashed{\c}{}_a,\mf{f}_a)$, namely the one compatible with $J$. 
Recalling the expression for the Lee form $\mf{f}_a=-\frac{1}{2}J_{b}{}^{c}\c_{c}J_{a}{}^{b}$, 
in terms of spin coefficients we have
\begin{equation}\label{f}
 \mf{f}_a=\rho n_a+\rho'\ell_a-\t\tilde{m}_a-\t' m_a,
\end{equation}
where we have chosen an arbitrary primed spin frame $\ve^{A'}_{\bf A'}=(o^{A'},\iota^{A'})$ for the primed 
spin bundle $\mc{S}^{A'}$ and defined the associated (complex) null tetrad in the usual way, i.e.
\begin{equation}\label{nt}
 \ell^a=\xi^Ao^{A'}, \quad n^a=\eta^{A}\iota^{A'}, \quad m^a=\xi^A\iota^{A'}, \quad \tilde{m}^a=\eta^Ao^{A'}.
\end{equation}

\subsubsection{The connection on spinor bundles induced from $\T$}

We have just seen that the canonical complex structure (\ref{J}) determines 
a preferred Weyl connection for the conformal manifold.
As mentioned, the fixing of the complex structure reduces $\SL_L$ to $\mbb{C}^{\times}$, 
which gives a subbundle $Q$ of $P_{\rm Spin}$
with structure group $\mbb{C}^{\times}\times\SL_R\times\mbb{R}^{+}$. (Recall that here 
$\mbb{R}^{+}$ is the multiplicative group of positive real numbers.)
From now on we choose the conformal weights for the spin frame $(\xi^A,\eta^A)$ as 
\begin{equation}\label{w0w1}
 w_0=0, \qquad w_1=-1.
\end{equation}

The principal bundle $Q$ inherits a connection from this reduction, which, since the Weyl 
connection is complex, will be valued in the 
complexified Lie algebra $\mf{g}_o:=(\mbb{C}\oplus\mf{sl}(2,\mbb{C})_R\oplus\mbb{R})\otimes\mbb{C}$.
This connection is found by looking at what parts of the full connection do not transform 
covariantly under the reduced structure group. 
A calculation similar to the one performed in \cite[Section 2.4]{Araneda2018}
shows that this connection 
is given by $\psi_a=(\w_a+B_a,\slashed{\w}{}_{a {\bf B'}}{}^{\bf C'},\mf{f}_a)$, where
$\slashed{\w}{}_{a {\bf B'}}{}^{\bf C'}=\ve^{\bf C'}_{B}\slashed{\c}{}_{a}\ve^{B}_{\bf B'}$ 
(with $\ve^{\bf A'}_{A}$ the frame dual to $\ve^{A}_{\bf A'}$) and
\begin{equation}
 \w_a:=-\e n_a+\e'\ell_a+\b\tilde{m}_a-\b' m_a, \qquad B_a:=-\rho n_a+\t\tilde{m}_a.
\end{equation}
$\w_a$ is the usual GHP connection form, and the 1-form $B_a$ was originally considered in \cite{Aks11} 
(for a choice of conformal weights different to (\ref{w0w1}) $B_a$ has to be modified, for details 
see \cite{Araneda2018}).

Now consider a section $\Psi\in\G(\mc{S}^{k,0}_{l,0})$, and project its indices on the frame 
$(\xi^A,\eta^A)$ and its dual, so that one gets a bunch of components.
A generic component $\psi$ is a complex scalar field that, under the allowed transformations of frame, 
changes according to a representation $\varrho_{p,w}$ of $\mbb{C}^{\times}\times\SL_R\times\mbb{R}^{+}$ 
on $\mbb{C}$ given by
\begin{equation}\label{rep2}
 \varrho_{p,w}(z,\tilde{S},\Om)\psi=z^p \Om^w \psi
\end{equation}
for some $p\in\mbb{Z}$. The scalar $\psi$ can then be regarded as a section of the complex line bundle
\begin{equation}\label{ws}
 \mc{O}(p)[w]:=Q\times_{\varrho_{p,w}}\mbb{C}.
\end{equation}
(In the language of the usual GHP formalism, sections of (\ref{ws}) could be thought of as 
`type $\{p,0\}$ quantities' with conformal weight $w$.)
Note that, if $\la \xi^A \ra$ is the line bundle whose fibre over $x\in\M$ is the set of spinors at $x$
proportional to $\xi^A$, we could also think of sections of (\ref{ws}) as complex-valued functions 
on $\la\xi^A\ra$ (namely $\psi:\la\xi^A\ra\to \mbb{C}$) that are homogeneous in $\xi^A$.

The connection on (\ref{ws}) is induced from the connection 1-form in $Q$ that we found before, 
and, using (\ref{IC}) and (\ref{rep2}), it is given by
\begin{equation}
 (\p_a+w\mf{f}_a+p(\w_a+B_a))\psi.
\end{equation}
More generally, for a section $\Psi\in\G(\mc{S}^{m,k'}_{n,l'})$, if we project an arbitrary number of 
its unprimed indices in the frame $(\xi^A,\eta^A)$, we get a mixed object that can be considered as 
a section of the product bundle $\mc{S}^{k,k'}_{l,l'}(p)[w]:=\mc{S}^{k,k'}_{l,l'}\otimes\mc{O}(p)[w]$
(for which we will also use the notation $\mc{S}^{A...A'...}_{B...B'...}(p)[w])$.
The connection on this structure is the product between the connections on the factors, so 
after all this discussion we finally get to:
\begin{lem}\label{lem-IC}
The 2D twistor space $\T$ from proposition \ref{prop-PR} induces a natural connection on the 
spinor bundles $\mc{S}^{k,k'}_{l,l'}(p)[w]$, given by
\begin{equation}\label{C}
\C_{a}\Psi^{B...B'...}_{C...C'...} := (\slashed{\c}{}_a+w\mf{f}_a+p(\w_a+B_a))\Psi^{B...B'...}_{C...C'...},
\end{equation}
where $\Psi^{B...B'...}_{C...C'...}\in\G(\mc{S}^{k,k'}_{l,l'}(p)[w])$.
\end{lem}

Summarizing, we have shown that the existence of a 2D twistor space defines in a 
natural way a preferred connection (\ref{C}) for the spinor bundles of the conformal structure.
The derivation is actually valid even if $\xi^A$ is not an SFR; the point is that the 2D twistor 
space singles out the (projective) spinor $\xi^A$. 
We can already see that the SFR condition is quite special, by noting that, 
since $\xi^A\in\G(\mc{S}^{A}(1)[0])$,
in terms of spin coefficients we have (see \cite[Eq. (2.53)]{Araneda2018}) 
\begin{equation}\label{derxi}
 \C_a\xi^B=(-\k n_a+\sigma\tilde{m}_a)\eta^B.
\end{equation}
The SFR condition on $\xi^A$ is equivalent to (\ref{SFR2}), so $\xi^A$ is in this case
annihilated by the naturally induced connection.

\subsection{Fibre bundles over the 2D twistor space $\T$}\label{sec-FB}

In section \ref{sec-PT} we have seen that the Penrose transform associates massless fields in a SD 
background spacetime with sheaf cohomology classes over twistor space. These cohomology classes
are sections of certain line bundles over $\PTc^*$ (modulo coboundary equivalence), 
that can be thought of as functions on the spin bundle that are constant on $\b$-surfaces 
(see discussion around (\ref{twf0})).
In order to study whether a similar mechanism can be constructed in our present context, in which
we do not have the full twistor space $\PTc^{*}$ but just the 2D twistor space $\T$, 
we have to construct bundles over $\T$. 
Recall that a single point $W\in\T$ corresponds to a whole 2-surface $\wt{W}$ in $\M$, so, 
roughly speaking, the construction of a fibre over $W$ would require objects that are 
appropriately `constant' over $\wt{W}$ (as in the case with a full twistor space).
This constancy will be expressed in terms of the connection $\C_a$ constructed before, 
and naturally it is constrained by integrability conditions
involving the curvature of $\C_a$, therefore we will first study this curvature.

\subsubsection{Curvature of $\C_a$}\label{sec-curv}

As usual, the curvature of the connection $\C_a$ is defined by the commutator $[\C_a,\C_b]$. 
This splits into its SD and ASD parts according to
\begin{equation}\label{comm}
 [\C_a,\C_b]=\e_{AB}\Box^{\C}_{A'B'}+\e_{A'B'}\Box^{\C}_{AB}
\end{equation}
where $\Box^{\C}_{A'B'}:=\C_{A(A'}\C_{B')}{}^{A}$ and $\Box^{\C}_{AB}:=\C_{A'(A}\C_{B)}{}^{A'}$. 
The irreducible decomposition of the second order operator $\C_{A'A}\C_{B}{}^{A'}$ is
\begin{equation}\label{2op}
 \C_{A'A}\C_{B}{}^{A'}=\tfrac{1}{2}\e_{AB}\Box^{\C}+\Box^{\C}_{AB}, \qquad \Box^{\C}:=g^{ab}\C_a\C_b.
\end{equation}
(Similarly for $\C_{A'A}\C_{B'}{}^{A}$.)
The ASD part of the curvature is $\Box^{\C}_{AB}$, and explicit expressions for it 
depend on the object it is acting on. We will focus on its action on sections of 
$\mc{O}(p)[w]$, $\mc{S}^{A'}(p)[w]$ and $\mc{S}^{A}(p)[w]$:

\begin{lem}
Let $f\in\G(\mc{O}(p)[w])$, $\m^{A'}\in\G(\mc{S}^{A'}(p)[w])$ and $\k^{A}\in\G(\mc{S}^{A}(p)[w])$. 
Suppose that $\xi^A$ is an SFR and a two-fold PND. Then
\begin{align}
 & \Box^{\C}_{AB}f= [ -(w\chi'+p\Psi_3) \xi_A\xi_B +3p\Psi_2 \xi_{(A}\eta_{B)}]f, \label{asdC} \\
 & \Box^{\C}_{AB}\m^{D'}= \Box_{AB}\m^{D'}+[ -(w\chi'+p\Psi_3) \xi_A\xi_B +3p\Psi_2 \xi_{(A}\eta_{B)}]\m^{D'}
 +F_{ABC'}{}^{D'}\m^{D'}, \label{asdC2}\\
 & \Box^{\C}_{AB}\k^{D}= \Box_{AB}\k^D+[ -(w\chi'+p\Psi_3) \xi_A\xi_B +3p\Psi_2 \xi_{(A}\eta_{B)}]\k^D
 +G_{ABC}{}^{D}\k^{D} \label{asdC3}
\end{align}
where we defined $\chi'=(\tho+2\rho-\tilde\rho)\k'-(\edt+2\t-\tilde\t')\sigma'+2\Psi_3$, 
$F_{ABC'}{}^{D'}=-\c_{(A}{}^{D'}\mf{f}_{B)C'}+\mf{f}_{(A}{}^{D'}\mf{f}_{B)C'}$, and 
$G_{ABC}{}^{D}=\e_{(B}{}^{D}[\c_{A)A'}f_{C}{}^{A'}+f_{A)A'}f_{C}{}^{A'}]$.
\end{lem}

\begin{proof}
From the definition (\ref{comm}), we have
\begin{equation*}
 \Box^{\C}_{AB}f=[w\c_{A'(A}\mf{f}_{B)}{}^{A'}+p( \c_{A'(A}\w_{B)}{}^{A'}+\c_{A'(A}B_{B)}{}^{A'} \;)]f.
\end{equation*}
The calculation of the RHS is tedious but straightforward, it can be done using the GHP formalism.
For an arbitrary spacetime, we find
\begin{align*}
 \c_{A'(A}\mf{f}_{B)}{}^{A'} &= -\chi' \xi_A\xi_B+\chi\eta_A\eta_B \\
 \c_{A'(A}\w_{B)}{}^{A'}+\c_{A'(A}B_{B)}{}^{A'} &= -\Psi_3\xi_A\xi_B
 +(3\Psi_2-2\zeta)\xi_{(A}\eta_{B)}-(\chi+\Psi_1)\eta_A\eta_B
\end{align*}
where $\zeta=\sigma\sigma'-\k\k'$ and $\chi$ is the GHP prime of $\chi'$. If $\xi^A$ 
is an SFR and a two-fold PND, then $\k=\sigma=0=\Psi_0=\Psi_1$, which implies 
$\chi=0=\zeta$ and (\ref{asdC}) follows. The proof of (\ref{asdC2}) and (\ref{asdC3}) is similar.
\end{proof}

Identities (\ref{asdC}) and (\ref{asdC2}) will be very useful below when studying the integrability 
conditions for differential equations associated to the construction of bundles over $\T$.

\subsubsection{The connection on $\b$-surfaces}

Consider an arbitrary $\b$-surface $\wt{W}$.
By definition, any tangent vector to $\wt{W}$ is of the form $\xi^A\m^{A'}$, with 
$\xi^A$ fixed and $\m^{A'}$ variable, thus the tangent bundle of $\wt{W}$, denoted $T\wt{W}$, 
can be identified with the primed spin bundle (more precisely, with the 
restriction of it to the $\b$-surface $\wt{W}$). 
We can also be more general and consider spinor fields with non-trivial $p$- and $w$-weights, 
by tensoring the corresponding bundle with $\mc{O}(p)[w]$.
Now, we have seen that the natural connection on the tangent bundle $T\M$, induced 
from the 2D twistor space, is $\C_a$.
To find the natural connection on $T\wt{W}$, we note that, for arbitrary $X,Y\in T\wt{W}$, 
this connection must satisfy $\C_{X}Y=Z$ for some $Z\in T\wt{W}$. If $X^a=\xi^A\pi^{A'}$, 
$Y^a=\xi^A\m^{A'}$
and $Z^a=\xi^A\zeta^{A'}$, this is equivalent to $\xi^{A}\pi^{A'}\C_{AA'}(\xi^B\m^{B'})=\xi^B\zeta^{B'}$.
Noting that this must be valid for arbitrary $\pi^{A'}$, contracting with $\eta_B$, 
and recalling that the right hand side should be a linear operator on $\m^{A'}$ satisfying 
the Leibniz rule, we get $-\eta_B\xi^{A}\C_{AA'}(\xi^{B}\m^{B'})\equiv\tilde{\C}_{A'}\m^{B'}$, 
defining in this way the natural connection $\tilde{\C}_{A'}$ on $T\wt{W}$ (see \cite{Bailey2} 
for similar discussion). (The notation $\tilde{\C}_{A'}$ instead of $\C_{A'}$ is chosen to match the 
conventions in section \ref{sec-hvb} below, where this is interpreted in terms of holomorphic structures.)
Furthermore, we have seen that the fact that $\xi^A$ is associated to a $\b$-surface implies that 
$\C_a\xi^B=0$, therefore
\begin{equation}\label{cbs}
 \tilde{\C}_{A'}\m^{B'}\equiv\xi^A\C_{AA'}\m^{B'}.
\end{equation}
An interesting result concerning this connection is the following:
\begin{lem}\label{lem-Cbs}
Suppose $\xi^A$ is an SFR and a two-fold PND. Then the connection on $\b$-surfaces is flat:
\begin{equation}
 [\tilde{\C}_{A'},\tilde{\C}_{B'}]=0.
\end{equation}
\end{lem}

\begin{proof}
We first note that 
\begin{equation*}
 [\tilde{\C}_{A'},\tilde{\C}_{B'}] = \xi^{A}\C_{A'A}(\xi^{B}\C_{B'B})-\xi^{B}\C_{B'B}(\xi^{A}\C_{A'A})
 =\e_{A'B'}\xi^A\xi^B\Box^{\C}_{AB}.
\end{equation*}
Now let $\m^{A'}$ be a section of $\mc{S}^{A'}(p)[w]$. Using (\ref{asdC2}) and the standard expression 
for the usual curvature operator $\Box_{AB}$, we get
$\xi^A\xi^B\Box^{\C}_{AB}\m^{D'}=(\xi^A\xi^B\Phi_{ABC'}{}^{D'}+\xi^A\xi^B F_{ABC'}{}^{D'})\m^{C'}$.
A straightforward but tedious calculation using the Ricci identities (see \cite[Eq. (4.12.32)]{Penrose1}) 
shows that in an arbitrary spacetime, introducing a primed spin frame $(o_{A'},\iota_{A'})$,
one has
\begin{align*}
 \xi^A\xi^B F_{ABC'D'}=&-(\edt'\k+\sigma\tilde\sigma+\Phi_{00})\iota_{C'}\iota_{D'}
 -(\tho'\sigma+\k\tilde\k'+\Phi_{02})o_{C'}o_{D'}\\
 & +[(\tho'+\rho')\k+(\tilde\t-\t')\sigma+\Psi_1+\Phi_{01}]\iota_{C'}o_{D'}\\
  & +[(\edt'+\t')\sigma+(\tilde\rho'-\rho')\k-\Psi_1+\Phi_{01}]o_{C'}\iota_{D'}
\end{align*}
If $\xi^A$ is an SFR and a two-fold PND, then this reduces to $\xi^A\xi^B F_{ABC'D'}=-\xi^A\xi^B\Phi_{ABC'D'}$, thus 
the result follows.
\end{proof}

The flatness of the connection $\tilde{\C}_{A'}$ has a number of interesting consequences, on 
which we will now comment only briefly. (We will not pursue these matters further here, and leave 
a more detailed analysis for future works).

First, consider an arbitrary $\b$-surface $\wt{W}$, and
denote by $V^{p,w}_k$ the sheaf of totally antisymmetric sections of the spinor bundle 
$\mc{S}_{A'...K'}(p)[w]$ (restricted to $\wt{W}$) with $k$ indices (which of course is zero for $k\geq3$). 
For a section $\psi_{A'...K'}=\psi_{[A'...K']}$, define the exterior derivative 
$({\rm d}_{\C}\psi)_{A'B'...L'}:=(k+1)\tilde{\C}_{[A'}\psi_{B'...L']}$. Then, since $\tilde{\C}_{A'}$ is flat, we 
have ${\rm d}_{\C}^2=0$, thus {\em we get a twisted de Rham complex}:
\begin{equation}\label{tdR}
\begin{tikzcd}
    0 \arrow{r} & V^{p,w}_0 \arrow{r}{{\rm d}_{\C}} & V^{p+1,w}_1 \arrow{r}{{\rm d}_{\C}} & V^{p+2,w}_2 \arrow{r} & 0.
\end{tikzcd} 
\end{equation}
Furthermore, a twisted de Rham complex is {\em locally exact} 
(see e.g. \cite[Prop. 2]{Kha14}, its proof, and references therein), meaning that for every point $x$ 
and for every function $f$ defined on a neighbourhood $U\ni x$ such that ${\rm d}_{\C}f=0$, there exists 
a function $g$ defined on $V\ni x$ (with $V\subseteq U$) such that $f={\rm d}_{\C}g$.
Now, recall that the exactness of a sequence of sheaves is a local requirement since it is at the 
level of stalks (see \cite[Def. 2.5 in Ch. II]{Wells}). 
More precisely, given three sheaves $\mc{A}$, $\mc{B}$, $\mc{C}$ over a topological space $X$, 
and two morphisms $\mc{A}\xrightarrow{\phi}\mc{B}$ and $\mc{B}\xrightarrow{\psi}\mc{C}$,
the sequence $\mc{A}\xrightarrow{\phi}\mc{B}\xrightarrow{\psi}\mc{C}$ 
is exact at $\mc{B}$ if the induced sequence on stalks,
$\mc{A}_x\xrightarrow{\phi_x}\mc{B}_x\xrightarrow{\psi_x}\mc{C}_x$, is exact at $\mc{B}_x$, 
namely ${\rm im}(\phi_x)=\ker(\psi_x)$ for all $x\in X$. 
Then one says that $0\to\mc{A}\xrightarrow{\phi}\mc{B}\xrightarrow{\psi}\mc{C}\to0$ is a short exact 
sequence of sheaves if it is exact at $\mc{A}$, $\mc{B}$ and $\mc{C}$ (namely, $\phi_x$ is injective, 
$\psi_x$ is surjective, and ${\rm im}(\phi_x)=\ker(\psi_x)$ for all $x\in X$).
Therefore, local exactness of a twisted de Rham complex implies that 
{\em (\ref{tdR}) is actually a short exact sequence of sheaves}.

Second, the fact that $\tilde{\C}_{A'}$ is flat implies that the equation 
\begin{equation}\label{Cmu}
 \tilde{\C}_{A'}\m^{B'}=0
\end{equation}
admits non-trivial solutions. This can be formulated in a way closer to the 
theory of integrable systems\footnote{I am grateful to J. L. Jaramillo for suggesting looking into this.}. 
More precisely, consider a primed spin frame $(o^{A'},\iota^{A'})$ such that $\tilde{\C}_{A'}o^{B'}=0$ 
and $\tilde{\C}_{A'}\iota^{B'}=0$, and introduce the following operators acting on $\G(\mc{S}^{A'}(p)[w])$:
\begin{equation}
 L:=\tilde{\C}_{0'}=o^{A'}\xi^{A}\C_{A'A}, \qquad  M:=\tilde{\C}_{1'}=\iota^{A'}\xi^{A}\C_{A'A}.
\end{equation}
Then (\ref{Cmu}) adopts the form of an overdetermined linear system
\begin{equation}
 L\m^{A'}=0, \qquad M\m^{A'}=0.
\end{equation}
The compatibility condition for this system is that the operators $L$ and $M$ must commute.
From their definition we have
\begin{equation}
 [L,M]=\xi^A\xi^{B}\Box^{\C}_{AB},
\end{equation}
therefore, the commutativity of $L$ and $M$ is equivalent to the flatness of the connection (\ref{cbs}). 
Formally, we can think of $L,M$ as a Lax pair, see e.g. \cite{Mason96} and \cite{Dun10}.

\smallskip
Finally, a particular application of equation (\ref{Cmu}) is that their solutions constitute the tangent bundle
$T\T$ to the 2D twistor space. This can be seen by adapting the discussion of Bailey in \cite{Bailey, Bailey2} 
to our context. (See also \cite[Section 9.1]{Ward}, which uses a local twistor description.)

\subsubsection{Complex line bundles}\label{sec-lb}

We now turn to the construction of line bundles over $\T$. Let $W$ be a point in $\T$, and 
$\wt{W}$ the corresponding $\b$-surface in the spacetime. 
Consider the restriction of the bundle $\mc{O}(p)[w]$ to $\wt{W}$, and let $f$ be 
a section of this bundle.
Different points on the $\b$-surface $\wt{W}$ correspond, by definition, to the same point
$W\in\T$, so in order to define a fibre over $W$ we require $f$ to be 
covariantly constant over $\wt{W}$, namely
\begin{equation}
 \C_{X}f=0
\end{equation}
for all $X$ tangent to $\wt{W}$, or equivalently
\begin{equation}\label{ccs}
 \xi^A\C_{AA'}f=0.
\end{equation}
Compare to (\ref{twf0}), (\ref{twf}).
Now, the spinor $\m^{A'}=\xi^A\C_{A}{}^{A'}f \;(=0)$ can be regarded as a (weighted) element of the tangent 
bundle $T\wt{W}$, for which the connection is (\ref{cbs}), therefore, the integrability conditions 
for (\ref{ccs}) on the $\b$-surface $\wt{W}$ can be obtained by applying an extra derivative $\tilde{\C}_{B'}$ 
and taking the commutator, which yields
\begin{equation}
 \xi^A\xi^B\Box^{\C}_{AB}f=0.
\end{equation}
If $\xi^A$ is an SFR and a two-fold PND, then these integrability conditions are satisfied by virtue of 
(\ref{asdC}), thus (\ref{ccs}) is a non-trivial condition.
We then use this fact to construct a line bundle over $\T$, by defining the fibre over a point $W$ to be 
composed of sections of $\mc{O}(p)[w]$ that satisfy (\ref{ccs}). 
This bundle will be denoted by $\mc{O}_{\T}(p)[w]$. The construction generalizes 
the one in SD spacetimes (which is needed for the Penrose transform) that we reviewed in 
section \ref{sec-PT}, to our current situation. 
Below we will see that these bundles give solutions to the Teukolsky equations on the spacetime.

\subsection{Teukolsky equations and massless fields}\label{sec-te}

We will now show that the above twistor constructions are intimately related to the 
description of massless free fields propagating in curved, algebraically special spacetimes, 
and give a natural interpretation to the relation between this description and
the appearance of various twistor objects that are known in the literature. 
(See also remark \ref{rem-ef} below.)

\subsubsection{Teukolsky equations}

\begin{lem}\label{mthm}
Sections of the line bundles $\mc{O}_{\T}(p)[w]$ are automatically solutions 
of the Teukolsky equations for massless free fields in the (conformal) spacetime.
\end{lem}

\begin{proof}
Let $f$ be a section of the line bundle $\mc{O}_{\T}(p)[w]$ over $\T$,
then by definition it satisfies (\ref{ccs}). Applying an additional 
derivative, we get $0=\C_{B}{}^{A'}(\xi^A\C_{AA'}f)$. By virtue of (\ref{derxi}),
the $\xi^A$ factor can be commuted to the left. Using then identities (\ref{2op}) and (\ref{asdC}), 
we obtain
\begin{equation}\label{gwe}
 (\Box^{\C}+3p\Psi_2)f=0.
\end{equation}
All we need to show now is that this is essentially the Teukolsky equation.
To this end, we express the wave operator $\Box^{\C}$ acting on $\mc{O}(p)[w]$ in 
GHP form. In an arbitrary spacetime, after some lengthy calculations we get
\begin{align}
\nonumber \Box^{\C} =& 2[(\tho'+p\rho'-\tilde\rho')(\tho-\rho)-(\edt'+p\t'-\tilde\t)(\edt-\t)]\\
\nonumber &-2(p-(w+1))(\rho'\tho+\rho\tho'-\t'\edt-\t\edt'-2\L-\Psi_2) \\
& +2(p-(w+1)(p-w))(\rho\rho'-\t\t')-3p\Psi_2-2(\k\k'-\sigma\sigma'). \label{sbox}
\end{align}
Note that the last term, i.e. $(\k\k'-\sigma\sigma')$, is zero if $\xi^A$ is an SFR.
Now set $p=-n$, $w=-n-1$ for $n\in\mbb{N}$, i.e. $f$ is a section of $\mc{O}_{\T}(-n)[-n-1]$. Then
\begin{align}
\nonumber 0=& \;(\Box^{\C}-3n\Psi_2)f\\
  =& \; 2[(\tho'-n\rho'-\tilde\rho')(\tho-\rho)-(\edt'-n\t'-\tilde\t)(\edt-\t)]f, \label{teuk1}
\end{align}
where the zero in the LHS of the first line is a consequence of (\ref{gwe}), and in the second line we 
have simply replaced (\ref{sbox}). But the second line is exactly the Teukolsky equation for 
the spin-weight $-n$ component of a massless free field with spin $n/2$, as presented for 
example in \cite{Teu73}.
\end{proof}

\begin{remark}\label{rem-teq}
Note that the wave equation (\ref{gwe}) is the natural generalization of (\ref{we}), 
see remark \ref{rem-we}.
In that case $f$ was a representative of a \u{C}ech cohomology class in $\breve{H}^1(\CP^1,\mc{O}(2h-2))$, 
which can be thought of as a function on the spin bundle that is homogeneous in the spinor variables, 
satisfies (\ref{twf}), and is subject to coboundary equivalence. 
In our present situation, $f$ is a function on the line bundle $\la\xi^A\ra$
that is homogeneous in the spinor variables and satisfies 
the generalized equation (\ref{ccs}), but we do not have a cohomological interpretation of it.
\end{remark}

\subsubsection{Massless free fields}\label{sec-mff}

Let us now briefly examine how to obtain massless free fields from the constructions above.
We start by considering LH fields.
First, note that if $\vp_{A...L}$ is a totally symmetric section of $\mc{S}_{A...L}(0)[-1]$ 
(that is, a symmetric spinor field with $p=0$ and $w=-1$), then 
\begin{equation}\label{cC}
 \c^{AA'}\vp_{A...L}=\C^{AA'}\vp_{A...L}.
\end{equation}
Now let for example $f$ be a section of $\mc{O}_{\T}(-1)[-2]$. Using (\ref{cC}) it follows immediately that 
$\vp_A(x,\xi)=\xi_A f(x,\xi)$ is a LH Dirac field, $\c^{AA'}\vp_A=0$.
Similarly, for a section $f$ of $\mc{O}_{\T}(-n)[-n-1]$, the spinor
\begin{equation}\label{TLH2}
 \vp_{A...L}(x,\xi)=\xi_{A}...\xi_{L} f(x,\xi)
\end{equation}
(with $n$ factors of $\xi_A$) is a LH massless field with spin $n/2$, $\c^{AA'}\vp_{A...L}=0$.

\begin{remark}\label{rem-LHF}
The field (\ref{TLH2}) is the generalization of (\ref{TLH}) to our present situation. 
Notice that there are no problems with Buchdahl constraints since by assumption $\xi^A$ is 
a two-fold PND of the ASD curvature.
\end{remark}
The result above is not really new, it is actually an expression of the `Robinson theorem' (see e.g. 
\cite[Theorem (7.3.14)]{Penrose2} and \cite{Bailey, BaileyTN262}), adapted to our constructions.

\smallskip
Let us now examine RH fields. Contrary to the LH case, now we will use sections of $\mc{O}(p)[w]$ 
that are {\em not} also sections of $\mc{O}_{\T}(p)[w]$, i.e. they do not satisfy (\ref{ccs}).
(In the language of section \ref{sec-TT}, we will use functions on $\la \xi^A \ra$ that 
``do not descend'' to $\T$.)
Let $h$ be a section of $\mc{O}(-1)[-1]$, and consider the spinor field
\begin{equation}\label{RHD2}
 \phi_{A'}:=\xi^A\C_{AA'}h(x,\xi).
\end{equation}
Using (\ref{cC}), (\ref{derxi}), (\ref{2op}) and (\ref{asdC}), we get
\begin{equation}\label{idRHD2}
 \c^{AA'}\phi_{A'}=\tfrac{1}{2}\xi^A(\Box^{\C}-3\Psi_2)h,
\end{equation}
thus, (\ref{RHD2}) is a RH Dirac field if and only if $h$ is a solution of the wave equation 
$(\Box^{\C}-3\Psi_2)h=0$. 
It is important to emphasize here that, since we have chosen the weights of $h$ as 
$p_h=-1$ and $w_h=-1$ (which are needed in order for $\phi_{A'}$ in (\ref{RHD2}) to have 
the correct weights, namely $p_{\phi}=0$ and $w_{\phi}=-1$), this wave equation {\em is not} 
the Teukolsky equation given in (\ref{teuk1}). 
In order to get the Teukolsky equation (\ref{teuk1}), we need a conformal factor $\mr\Om$ 
(i.e. an element of $\G(\mc{O}(0)[1])$) such that $\C_a\mr\Om=0$, thus the field $\tilde{h}=\mr\Om^{-1}h$ 
has weights $p_{\tilde{h}}=-1$ and $w_{\tilde{h}}=-2$ and consequently satisfies (\ref{teuk1}) (for $n=1$) 
provided that $h$ satisfies $(\Box^{\C}-3\Psi_2)h=0$. The requirement $\C_a\mr\Om=0$ 
is a non-trivial condition, see section \ref{sec-2TM} below (especially remark \ref{rem-ef}).

\begin{remark}\label{rem-RHD}
The field (\ref{RHD2}) is the generalization of (\ref{RHD}) to our case (note their analogous structure). 
In the case of (\ref{RHD}), it satisfies the Dirac equation because $h_i$ satisfies the wave equation, 
which in turn is a consequence of the fact that $\Box h_i$ is a global function in $\CP^1$, homogeneous of 
degree $-1$. In other words, $\Box h_i=0$ is automatic from the structure of the cohomology groups 
involved in the construction. In our current situation it seems that we do not have enough structure 
to do cohomology\footnote{Note that, roughly speaking, a fibre of $\la\xi^A\ra$ corresponds to a 
single point in a fibre of $\mbb{P}\mc{S}^A$.}, 
so we were not able to give a cohomological interpretation to $h$ in (\ref{RHD2}).
\end{remark}

Consider now RH Maxwell fields. For this case we find it more convenient to propose the 
Ans\"atz $A_{AA'}=\xi_A\xi^B\C_{BA'}h$, where $h\in\G(\mc{O}(-2)[-1])$. 
The LH and RH parts of the 2-form $F_{ab}=\c_{[a}A_{b]}$ are respectively
$\psi_{AB}=\c_{A'(A}A_{B)}{}^{A'}$ and $\phi_{A'B'}=\c_{A(A'}A_{B')}{}^{A}$.
The vector potential $A_{AA'}$ has weights $p=0$ and $w=0$, and one can show that this implies
that we can replace $\c_a$ by $\C_a$ in the formulas for $\psi_{AB}$ and $\phi_{A'B'}$.
An easy calculation using (\ref{derxi}), (\ref{2op}) and (\ref{asdC}) leads to
\begin{align}
 \psi_{AB} &=-\tfrac{1}{2}\xi_A\xi_B(\Box^{\C}-6\Psi_2)h, \label{LHM2} \\
 \phi_{A'B'} &=\xi^A\xi^B\C_{AA'}\C_{BB'}h \label{RHM2}
\end{align}
(the symmetrization in $A'B'$ not being needed by virtue of (\ref{asdC})).
Therefore, $\psi_{AB}=0$ if and only if $h$ is a solution of the wave equation $(\Box^{\C}-6\Psi_2)h=0$, 
case in which the RH part $\phi_{A'B'}$ is a solution of the Maxwell equations.

\begin{remark}[Yang's equation]
Our procedure here turns out to be closely related to other approaches for 
solving the ASD Yang-Mills equations. 
That is, we are solving the equation $F_{ab}=\phi_{A'B'}\e_{AB}$, i.e., the ASD curvature 
of $A_{AA'}$ vanishes, $\psi_{AB}=0$. This is done by using a solution $h$ of the Teukolsky-like 
equation $(\Box^{\C}-6\Psi_2)h=0$. But we have seen in the proof of lemma \ref{mthm} that this 
equation is $\eta^A\C_{AA'}(\xi^{B}\C_{B}{}^{A'}h)=0$, which can be interpreted as a generalized 
version of what is known as {\em Yang's equation} (see pp. 165 in \cite{Mason96}, and also the 
discussion leading to eq. (3.1.3) in that reference).
\end{remark}

\subsubsection{Comments on gravitational perturbations and Hertz potentials}\label{sec-gp}

Gravitational perturbations of a curved spacetime cannot be described with the constructions 
above, for a number of reasons: $(i)$ the corresponding field equations are not the ones of a 
massless free field, $(ii)$ the Einstein equations are not conformally invariant, and 
$(iii)$ arbitrary perturbations in principle do not satisfy the conditions for admitting a 
2D twistor space.
Nevertheless, we find it useful to make some comments on this case and point out some 
interesting properties, especially regarding $(iii)$. (See also remark \ref{rem-l2dts} below.)

All of our constructions so far depend on the existence of a 2-dimensional twistor manifold. 
Even though this is much less restrictive than the existence of a twistor 3-manifold
(since the latter would imply SD curvature), the condition still singles out a particular 
class of spacetimes, namely (by proposition \ref{prop-PR}) those admitting an 
SFR (which we have also assumed to be a two-fold PND).
When perturbing (linearly) a spacetime, the metric becomes 
$g_{ab}+\ve h_{ab}$, and the property of having an SFR is generally destroyed by the 
perturbation, so one does not expect the constructions of the previous sections to 
apply to perturbed spacetimes.

Now, a particularly relevant method of generating metric perturbations is by the 
so-called Hertz/Debye potentials; this has been of interest both in past and recent years, 
see e.g. \cite{Kegeles, Aks16, Araneda2016, Pra18}.
A Hertz potential is a solution of higher spin field equations (Dirac, Maxwell, linearized gravity)
that is obtained by applying linear differential operators to a scalar field (so-called Debye potential)
that solves a certain scalar, wave-like equation
(for example, the field $h$ in (\ref{RHD2}) and (\ref{RHM2}) is a Debye potential).
These potentials are of much interest in the stability problem for black holes, since it 
is conjectured that all relevant gravitational perturbations can be generated this way 
(see e.g. \cite{Pra18, Aks16}). We will prove the following:

\begin{thm}\label{thm}
 Linearized metric perturbations generated by Hertz potentials possess a 2-dimensional twistor manifold.
\end{thm}

Below, in remark \ref{rem-l2dts}, we comment on the usefulness of this result and its possible applications.
We will prove theorem \ref{thm} by showing that the (linearized) ASD Weyl spinor of such perturbations 
is algebraically special (lemma \ref{lem-as} below), and that this implies that the perturbed 
spacetime still possesses a shear-free null geodesic congruence (lemma \ref{lem-lSFR} below).
That is, we obtain the linearized version of the Goldberg-Sachs 
theorem in one direction\footnote{Of course, the `if and only if' part of the linearized version 
of the Goldberg-Sachs theorem is not valid, as shown in \cite{Dain00}.}.

Let $f$ be a section of $\mc{O}(-4)[-5]$. From the expression (\ref{sbox}) for $\Box^{\C}$ 
and eq. (2.14) in \cite{Teu73}, one deduces that the Teukolsky equation for gravitational 
perturbations is
\begin{equation}\label{tegp}
 (\Box^{\C}-18\Psi_2)f=0.
\end{equation}
(Observe that (\ref{tegp}) is {\em not} a particular case of (\ref{gwe}), 
since the explicit form of the operator $\Box^{\C}$ depends on $p$.) 
Let $\xi_{ABCD}=\xi_A\xi_B\xi_C\xi_D$, $\xi_A$  being an SFR. 
One can show (see \cite{Kegeles, Araneda2016, Aks16})
that if $f$ is a solution of the Teukolsky equation (\ref{tegp}), then the tensor field
\begin{equation}\label{h}
 h_{AA'BB'}=\c_{(A'}{}^{C}[(\c_{B')}{}^{D}+4\mf{f}_{B')}{}^{D})\xi_{ABCD}f]
\end{equation}
is a solution of the linearized Einstein equations. We have

\begin{lem}\label{lem-as}
The linearized ASD Weyl spinor of the metric perturbation (\ref{h}) is algebraically 
special: $\xi^A$ is a two-fold PND.
\end{lem}

\begin{proof}
We have to prove that $\xi^B\xi^C\xi^D\dot\Psi_{ABCD}=0$, where $\dot\Psi_{ABCD}$ is the linearized 
ASD Weyl spinor of (\ref{h}).
A simple way to prove this is by considering a modified covariant derivative constructed from 
$\C_a$. Define $D_a=\c_a+p(\w_a+B_a)$ such that it acts on tensor/spinor fields with a $p$-weight.  
Note that for fields with $p=0$, $D_a$ coincides with the Levi-Civita derivative $\c_a$.
Letting $\xi^{A...K}=\xi^{A}...\xi^{K}$ ($n$ factors of $\xi^A$), equation (\ref{derxi}) implies 
the following two identities in terms of $D_a$:
\begin{align}
 & D_{A'}{}^{(A}\xi^{B....L)}=0, \label{wks} \\
 & (D_{A'A}+(n+1)\mf{f}_{A'A})\xi^{A...K}=0. \label{dwks}
\end{align}
Using (\ref{dwks}) for $n=4$ and the fact that $\xi_{ABCD}f$ has zero $p$-weight, 
we can write (\ref{h}) as
\begin{equation}\label{h2}
 h_{AA'BB'}=\xi_{ABCD}(D_{(A'}{}^{C}-5\mf{f}_{(A'}{}^{C})(D_{B')}{}^{D}-\mf{f}_{B')}{}^{D})f.
\end{equation}
Now, the linearized ASD Weyl spinor of a metric perturbation is given in general by (see 
\cite[Eq. (5.7.15)]{Penrose1})
\begin{equation*}
 \dot{\Psi}_{ABCD}=\tfrac{1}{2}\c_{(A}{}^{A'}\c_{B}{}^{B'}h_{CD)A'B'}+\tfrac{1}{4}g^{ef}h_{ef}\Psi_{ABCD}
\end{equation*}
Note that, since (\ref{h}) has zero $p$-weight, we can replace 
$\c_{a}$ by $D_{a}$ in this expression. Furthermore we have $g^{ef}h_{ef}=0$ for (\ref{h}), so we get
\begin{equation}\label{psih}
 \dot{\Psi}_{ABCD}=\tfrac{1}{2}D_{(A}{}^{A'}D_{B}{}^{B'}[\xi_{CD)}M_{A'B'}],
\end{equation}
where $M_{A'B'}=\xi_E\xi_F(D_{(A'}{}^{E}-5\mf{f}_{(A'}{}^{E})(D_{B')}{}^{F}-\mf{f}_{B')}{}^{F})f$.
Using (\ref{wks}), the $\xi_{CD}$ factor inside the bracket in (\ref{psih}) can be commuted to the 
left. Projecting then over $\xi^{BCD}$, it follows that $\xi^{BCD}\dot\Psi_{ABCD}=0$.
\end{proof}

We will now investigate the existence of $\b$-surfaces in the perturbed spacetime. 
Since the linearization of spinors is a subtle issue, we find it more clear to formulate
the discussion primarily in tensor terms. 
Consider a monoparametric family $g_{ab}(\ve)$ such that $\mr{g}_{ab}\equiv g_{ab}(0)$ is 
our background spacetime. Consider also four vector fields 
$\ell^a(\ve)$, $n^a(\ve)$, $m^a(\ve)$ and $\tilde{m}^a(\ve)$ that constitute a null tetrad 
for all values of the parameter $\ve$ 
(that is, $g_{ab}(\ve)\ell^a(\ve)n^b(\ve)=1=-g_{ab}(\ve)m^a(\ve)\tilde{m}^b(\ve)$ 
and all other products vanish). 
We assume all fields to depend smoothly on $\ve$, so that we have the Taylor expansions 
$\ell^a(\ve)=\mr{\ell}^a+\ve\dot{\ell}^a+O(\ve^2)$, $m^a(\ve)=\mr{m}^a+\ve\dot{m}^a+O(\ve^2)$, 
etc\footnote{In what follows, for a quantity $T(\ve)$ we use the notation $\mr{T}\equiv T(0)$
and $\dot{T}\equiv\frac{d}{d\ve}|_{\ve=0}T(\ve)$.}.
At any point $x\in\M$, the vector fields $\ell^a(\ve)$ and $m^a(\ve)$ generate a $\b$-plane. 
The condition for this $\b$-plane to be the tangent plane to a $\b$-surface 
is that the commutator of $\ell^a(\ve)$ and $m^a(\ve)$ should be a linear combination of them. 
Assuming that the background spacetime possesses an SFR (which implies $\mr\k=0=\mr\sigma$),
to linear order we find
\begin{equation}\label{lbs}
 [\mr{\ell}+\ve\dot{\ell},\mr{m}+\ve\dot{m}]^a=
 a(\mr{\ell}^a+\ve\dot{\ell}^a)+b(\mr{m}^a+\ve\dot{m}^a)+\ve(-\dot\k\mr{n}^a
 +\dot\sigma\mr{\tilde{m}}{}^{a})+O(\ve^2)
\end{equation}
for some scalar fields $a, b$. This means that, to linear order, the $\b$-surface condition 
is satisfied if and only if $\dot\k=0=\dot\sigma$. 

\begin{lem}\label{lem-lSFR}
Consider a background Einstein spacetime that possesses an SFR, and such that $\mr\Psi_2\neq0$. 
Consider also a perturbation $h_{ab}$ of this spacetime that satisfies the linearized Einstein 
vacuum equations (cosmological constant allowed), and let $\dot{\Psi}_{ABCD}$ be the linearized 
ASD Weyl curvature spinor of $h_{ab}$. If $\dot{\Psi}_{ABCD}$ is algebraically special 
along the background PND, then 
\begin{equation}\label{lSFR}
 \dot\k=0=\dot\sigma.
\end{equation}
\end{lem}

\begin{proof}
The proof is immediate by considering the following two Bianchi identities in GHP form, which are 
valid for an arbitrary spacetime:
\begin{align*}
 & (\tho-4\rho)\Psi_1-(\edt'-\t')\Psi_0-(\tho-2\tilde{\rho})\Phi_{01}+(\edt-\tilde{\t}')\Phi_{00}= 
  -3\k\Psi_2-2\sigma\Phi_{10}+2\k\Phi_{11}+\tilde{\k}\Phi_{02} \\
 & (\tho'-\rho')\Psi_0-(\edt-4\t)\Psi_1-(\edt-2\tilde{\t}')\Phi_{01}+(\tho'-\tilde{\rho}')\Phi_{02}= 
 3\sigma\Psi_2-2\k\Phi_{12}+2\sigma\Phi_{11}+\tilde{\sigma}\Phi_{00}.
\end{align*}
The Goldberg-Sachs theorem for the background solution implies that $\mr\k=\mr\sigma=0=\mr\Psi_0=\mr\Psi_1$, 
and the background Einstein equations are $\mr\Phi_{ab}=0$.
Linearizing the above Bianchi identities around the background solution, imposing the linearized Einstein 
equations $\dot\Phi_{ab}=0$, and the two-fold PND condition $\dot\Psi_0=\dot\Psi_1=0$, we get 
$\dot\k\mr\Psi_2=0$ and $\dot\sigma\mr\Psi_2=0$, which implies (\ref{lSFR}).
\end{proof}

It was shown in \cite{Dain00} that the Goldberg-Sachs theorem is not valid in linearized
gravity. More precisely, the results of \cite{Dain00} show that the linearized version of 
the Goldberg-Sachs theorem is not valid {\em in one direction}: the existence of an SFR in 
a perturbed spacetime does not imply that the corresponding linearized Weyl tensor is 
algebraically special. Lemma \ref{lem-lSFR} asserts that the converse is actually true.

\smallskip
Summarizing, from equation (\ref{lbs}) we see that, at the linearized level, the existence 
of $\b$-surfaces requires $\dot\k=0=\dot\sigma$, and from lemma \ref{lem-lSFR} we see that 
this condition is satisfied as long as $\dot\Psi_0=\dot\Psi_1=0$ (and the linearized Einstein 
equations hold too). Lemma \ref{lem-as} implies that metric perturbations generated by Hertz 
potentials satisfy these requirements, thus, we conclude that the perturbed spacetime admits 
$\b$-surfaces to linear order.

\begin{remark}\label{rem-l2dts}
The fact that the perturbed spacetime admits a 2D twistor space at the linearized level
has potentially interesting consequences. 
More precisely, complex spacetimes that admit totally null surfaces and that are 
half-algebraically special (i.e. such that one of the Weyl curvature spinors is algebraically special) 
are sometimes called `Hyperheaven spaces' and were studied thoroughly by Pleba\'nski 
and collaborators\footnote{I am very grateful to M. Dunajski and L. Mason
for discussions about this and for suggesting references.}, see \cite{Ple76, Fin76, Boy79}.
Lemmas \ref{lem-as} and \ref{lem-lSFR} suggest that our construction here may be a 
linearized version of the (non-linear) Hyperheaven construction, and that 
the Hertz potential (\ref{h}) and the Teukolsky equation (\ref{tegp}) might just be the 
linearized versions of the Hyperheaven metric reconstruction and the hyperheavenly equation 
respectively. Work on this is in progress \cite{Ara19}.
\end{remark}

\section{Spaces with two 2D twistor spaces}\label{sec-2TM}

\subsection{Preliminaries}

So far we have assumed the existence of a single 2D twistor manifold $\T$, which 
by proposition \ref{prop-PR} is equivalent to the existence of an SFR on a conformal structure.
The existence of two independent 2D twistor spaces, say $\T$ and $\T'$, implies that the 
conformal spacetime admits two families of null, geodesic, shear-free congruences. 
In this section we will analyse the case where the two 2D twistor spaces are associated to 
$\b$-surfaces.

Before studying this case in more detail, it is worth noting an interesting related 
construction\footnote{I thank M. Dunajski for bringing this reference to my attention.} 
\cite{Dun06} involving more than one 2D twistor space. 
Consider a conformal structure $(\M,[g])$ that admits a null conformal Killing vector, i.e. a 
vector field $K^a$ such that $\pounds_{K}g_{ab}\propto g_{ab}$ for any metric in the conformal class. 
Since $K^a$ is null we can write it as $K^a=\xi^A\m^{A'}$ for some spinor fields $\xi^A$ and $\m^{A'}$.
Then it is not difficult to show that the conformal Killing equation for $K^a$ implies 
$\xi^{A}\xi^{B}\c_{AA'}\xi_B=0$ and $\m^{A'}\m^{B'}\c_{AA'}\m_{B'}=0$, i.e. both $\xi^A$ and $\m^{A'}$ 
are geodesic and shear-free. Thus they give two independent foliations of $\M$ by totally null 2-surfaces: 
$\b$-surfaces for $\xi^A$ and $\a$-surfaces for $\m^{A'}$, and the moduli spaces of them give two 
kinds of 2D twistor spaces, that in usual twistor terminology would be `dual' to each other.
The $\a$- and $\b$- surfaces in the spacetime intersect along integral curves of $K^a$, which are null geodesics.
If, furthermore, the conformal structure is ASD, then $\a$-surfaces exist automatically and form a 3-dimensional 
twistor space $\PTc$ \cite{Pen76}, and the special $\a$-surfaces of $\m^{A'}$ give a hypersurface in $\PTc$. 
One can then ask what is the relationship between the space of $\b$-surfaces of $\xi^A$ (what we are here calling a 
2D twistor space $\T$) and the twistor space $\PTc$ of the ASD conformal structure;
this case was thoroughly studied in \cite{Dun06}, we refer to it for more details (see also \cite{Cal06}).

In this section we are interested, instead, in the existence of two 2D twistor spaces associated to foliations 
of $\M$ by $\b$-surfaces. For conformally Einstein spacetimes
this is naturally associated to Petrov type D spaces, by virtue of the Goldberg-Sachs theorem;
this case is particularly interesting because it includes the stationary black hole solutions 
(with or without cosmological constant).

\subsection{Holomorphic vector bundles}\label{sec-hvb}

From now on we consider two 2D twistor spaces, $\T$ and $\T'$, associated to the 
projective spinors $o^A$ and $\iota^A$ respectively.
Let us first see that the almost-complex structure (\ref{J}) in this case can be seen as 
induced naturally by the complex manifold structure of the $\b$-surfaces.
For each point $p$ in $\M$ there are two $\b$-surfaces, say $\Sigma$ and $\Sigma'$, 
passing through it. Suppose that $\Sigma$ has complex coordinates $(\tilde{w},\tilde{z})$ 
and tangents $\tilde{\p}_w,\tilde{\p}_z$, and $\Sigma'$ has complex coordinates $(w,z)$
and tangents $\p_w,\p_z$.
The point $p$ can then be given coordinates $(\tilde{w},\tilde{z},w,z)$, 
and the tangent space $T_p\M$ is spanned by $(\tilde{\p}_w,\tilde{\p}_z,\p_w,\p_z)$. 
Similarly, the cotangent space $T^{*}_p\M$ is spanned by the dual basis 
$({\rm d}\tilde{w}, {\rm d}\tilde{z},{\rm d}w, {\rm d}z)$.
The complex manifold structure induces an almost-complex structure $J:T_p\M\to T_p\M$ as usual: 
$J\p_w=+i\p_w$, $J\p_z=+i\p_z$, $J\tilde{\p}_w=-i\tilde{\p}_w$, $J\tilde{\p}_z=-i\tilde{\p}_z$.
In terms of these bases, we have the standard expression (see e.g. \cite[eq. (8.21)]{Nak03})
\begin{equation}\label{J2}
J=i{\rm d}z\otimes\p_z+i{\rm d}w\otimes\p_w-i{\rm d}\tilde{z}\otimes\tilde{\p}_z-i{\rm d}\tilde{w}\otimes\tilde{\p}_w.
\end{equation}
But from the basic definition of $\b$-surfaces, we know that the tangents to $\Sigma$ must have 
the form $\tilde{\p}_w=o^A\m^{A'}\p_{AA'}$ and $\tilde{\p}_z=o^A\n^{A'}\p_{AA'}$ for some $\m^{A'}$
and $\n^{A'}$, and likewise the tangents to $\Sigma'$ have to be $\p_{w}=\iota^A\m^{A'}\p_{AA'}$ 
and $\p_{z}=\iota^A\n^{A'}\p_{AA'}$. Choosing the normalization $o_A\iota^A=1=\m_{A'}\n^{A'}$, 
for the dual basis we get ${\rm d}\tilde{w}=\iota_A\n_{A'}{\rm d}x^{AA'}$, 
${\rm d}\tilde{z}=-\iota_A\m_{A'}{\rm d}x^{AA'}$, ${\rm d}w=-o_A\n_{A'}{\rm d}x^{AA'}$ and 
${\rm d}z=o_A\m_{A'}{\rm d}x^{AA'}$. Replacing then these expressions in (\ref{J2}), we easily find
$J=J_{AA'}{}^{BB'}{\rm d}x^{AA'}\otimes\p_{BB'}$, with the components $J_{AA'}{}^{BB'}$ given 
exactly by (\ref{J}) (recall that in this section we use $\xi^A\equiv o^A$ and $\eta^A\equiv \iota^A$).
This shows that, in the case we have two 2D twistor spaces, the almost-complex structure (\ref{J}) acquires
a natural interpretation as induced by the complex manifold structure of the foliations by $\b$-surfaces.

\begin{remark}
For the case with just one 2D twistor space studied in section \ref{sec-2DTM} (that is, with just one 
foliation of $\M$ by $\b$-surfaces), we can still define an almost-complex structure by 
$J\tilde{\p}_w=-i\tilde{\p}_w$, $J\tilde{\p}_z=-i\tilde{\p}_z$, and $Ju=iu$, $Jv=iv$ where 
$u,v$ are such that $(\tilde{\p}_w,\tilde{\p}_z,u,v)$ is a basis for the tangent space; this way we end 
up with (\ref{J}) again. But in such case this $J$ is not induced by a complex manifold structure, 
since it is not integrable.
\end{remark}

Now, the complex structure allows us to give a notion of {\em holomorphicity}. More precisely, the 
fact that the map $J$ has eigenvalues $+i,+i,-i,-i$ allows a decomposition of any tangent space 
as $T_p\M=T^{(1,0)}_p\M\oplus T^{(0,1)}_p\M$, where $T^{(1,0)}_p\M$ corresponds to the eigenvalue $+i$
and its elements are called holomorphic vectors, and $T^{(0,1)}_p\M$ corresponds to $-i$
and its elements are anti-holomorphic vectors. We can do the same for the cotangent space, 
and more generally we can decompose the bundle of $k$-forms into type $(r,s)$-forms in the usual
way, i.e. $\L^k\M=\bigoplus_{r+s=k}\L^{(r,s)}\M$. This allows us to introduce Dolbeault operators, 
which are a convenient way of capturing the notion of holomorphic fields (see e.g. \cite[Section 2.2]{Ada17} 
and \cite[Section 9.5]{Mason96}).
That is, introducing the projection to type $(r,s)$-forms $\pi_{r,s}:\L^{k}\M\to\L^{(r,s)}\M$, 
we define the Dolbeault operators $\tilde{\p}:=\pi_{r,s+1}\circ{\rm d}$ 
and $\p:=\pi_{r+1,s}\circ{\rm d}$. In terms of complex 
coordinates this is $\tilde{\p}={\rm d}\tilde{x}^{A'}\wedge\tilde{\p}_{A'}$ 
and $\p={\rm d}x^{A'}\wedge\p_{A'}$, where 
$\tilde{x}^{A'}=(\tilde{w},\tilde{z})$, $x^{A'}=(w,z)$ and $\tilde{\p}_{A'}=o^{A}\p_{AA'}$, 
$\p_{A'}=\iota^A\p_{AA'}$.
An ordinary scalar function is said to be holomorphic with respect to this complex 
structure if $\tilde{\p}f=0$.
But, as we have seen in section \ref{sec-2DTM}, we need sections of the line bundles 
$\mc{O}(p)[w]$. The operator $\tilde{\p}$ is not a connection in these bundles, since it does 
not map sections to sections. 
What we can do is combine this operator with the connection $\C_a$ and define a 
{\em partial connection} \cite[Section 9.5]{Mason96}
(or deformation of the complex structure, in the sense of e.g. \cite[Section 4.1]{Ada17})
as $\tilde{\p}_{\C}:={\rm d}\tilde{x}^{A'}\wedge\tilde{\C}_{A'}=\tilde{\p}+\tilde{a}$, 
where $\tilde{a}$ is a $\mf{g}_o$-valued type $(0,1)$-form given by 
$\tilde{a}_{A'}=o^A(wf_{AA'}+p(\w_{AA'}+B_{AA'}))$ (for the case of the bundle $\mc{O}(p)[w]$).

\begin{remark}
The bundle $\mc{O}(p)[w]$ equipped with $\tilde{\p}_{\C}$ is a holomorphic line bundle: 
$\tilde{\p}_{\C}^2=0$.
\end{remark}

\begin{proof}
This is an immediate consequence of the fact that the connection $\tilde{\C}_{A'}$ is flat, 
which was proven in lemma \ref{lem-Cbs}.
\end{proof}

The fact that the line bundle $\mc{O}(p)[w]$ is holomorphic allows to define the notion of 
a {\em holomorphic section} of it, as a section $f$ such that $\tilde{\p}_{\C}f=0$.
Explicitly, this is exactly the condition (\ref{ccs}), so the line bundle $\mc{O}_{\T}(p)[w]$ 
over the 2D twistor space $\T$ can be characterized as the bundle of holomorphic 
sections of $\mc{O}(p)[w]$ with respect to the complex structure defined above.

Similarly, we can define an {\em anti-holomorphic section} of $\mc{O}(p)[w]$ as a section $g$
that satisfies $\p_{\C}g=0$, which is equivalently $\iota^{A}\C_{AA'}g=0$. 
The integrability condition for this is $\p^2_{\C}=0$, or explicitly $\iota^A\iota^B\Box^{\C}_{AB}g=0$.
Calculations analogous to the ones in section \ref{sec-FB} show that this condition is satisfied if 
$\iota^A$ is an SFR (which is automatic since $\iota^A$ is associated to the 2D twistor space $\T'$) 
and a two-fold PND (which we will assume from now on).
Therefore we can construct the line bundles $\mc{O}_{\T'}(p)[w]$ over $\T'$, by using 
anti-holomorphic sections of $\mc{O}(p)[w]$.
These can be used to generate solutions of the Teukolsky equations with opposite spin-weight 
to the one given by the equations in lemma \ref{mthm}.
To see this, let $g\in\G(\mc{O}_{\T'}(p)[w])$; then by an analogous calculation to the one leading
to (\ref{gwe}), we now have
\begin{equation}
 (\Box^{\C}-3p\Psi_2)g=0.
\end{equation}
Now commute the GHP operators with their primed versions in (\ref{sbox}); after tedious 
calculations one gets
\begin{align}
\nonumber \Box^{\C} =& 2[(\tho-p\rho-\tilde\rho)(\tho'-\rho')-(\edt-p\t-\tilde\t')(\edt'-\t')]\\
\nonumber &+2(w+1)(\rho\tho'+\rho'\tho-\t\edt'-\t'\edt-2\L-\Psi_2) \\
& -2(w+1)(p-w)(\rho\rho'-\t\t')+3p\Psi_2+2(p-1)(\k\k'-\sigma\sigma').
\end{align}
(This expression is valid in an arbitrary spacetime.) Choosing $p=n$, $w=-1$, it follows that
\begin{align}
\nonumber 0=& \;(\Box^{\C}-3n\Psi_2)g\\
  =& \; 2[(\tho-n\rho-\tilde\rho)(\tho'-\rho')-(\edt-n\t-\tilde\t')(\edt'-\t')]g, \label{teuk2}
\end{align}
which is the Teukolsky equation for the spin-weight $+n$ component of a massless 
free field with spin $n/2$ (see \cite{Teu73}).

\begin{remark}
 In this context we can also generate solutions of the `Fackerell-Ipser equation', which is the wave-like 
 equation satisfied by the spin-weight zero component of a Maxwell field in a type D spacetime.
 Namely, if $g$ is a section of $\mc{O}_{\T'}(0)[-1]$, then 
 \begin{equation}
 0=\Box^{\C}g= 2[(\tho-\tilde\rho)(\tho'-\rho')-(\edt-\tilde\t')(\edt'-\t')]g, \label{FI}
 \end{equation}
 which, after noting that the differential operator on the RHS is
 $(\Box+2\Psi_2+4\L)$,  is exactly the Fackerell-Ipser equation. (Actually it is a generalized version 
 including the Ricci scalar ---recall that $\L=R/24$.)
\end{remark}

\subsection{Massless free fields and symmetry operators}

Finally, it is worth discussing the construction of RH massless free fields in our 
present situation. We start with the Dirac case.
If $h^{-}\in\G(\mc{O}(-1)[-1])$ and $h^{+}\in\G(\mc{O}(1)[0])$, and we define 
$\phi^{-}_{A'}=o^A\C_{AA'}h^{-}$ (which is simply (\ref{RHD2})) and 
$\phi^{+}_{A'}=\iota^A\C_{AA'}h^{+}$, then a calculation analogous to the one leading 
to (\ref{idRHD2}) shows that these fields are solutions of the RH Dirac equation if 
$h^{\pm}$ are solutions of the corresponding wave equations.
This process is particularly interesting in relation to {\em symmetry operators} (see 
e.g. \cite{Araneda2016, And14} and references therein), by which we mean the idea of 
applying differential operators to solutions of the LH massless field equations 
in such a way that one constructs solutions of the RH field equations.
Consider for example a LH Dirac field $\vp_A$. 
The first guess is to put $h^{-}=\vp_A\iota^A$ and $h^{+}=\vp_Ao^A$, but, since the conformal 
weights of $\vp_A$, $o^A$, and $\iota^A$ are respectively $-1$, $0$ and $-1$, the
$h^{\pm}$ so defined would not have the correct conformal weights. To remedy this situation,
we consider a conformal factor $\mr\Om$ (i.e. a section of $\mc{O}(0)[1]$) such that 
$\C_a\mr\Om=0$. This is only possible if there is a non-trivial solution to 
$\C_a\mr\Om=\p_a\mr\Om+\mf{f}_a\mr\Om=0$, namely if $\mf{f}_a=-\mr\Om^{-1}\p_a\mr\Om$. 
(This case is particularly interesting and deserves some additional comments, 
see remark \ref{rem-ef} below.)
This condition is satisfied for instance in all type D conformal structures that 
admit an Einstein metric, since then the Bianchi identities imply that 
$\mf{f}_a=\Psi^{-1/3}_2\p_a\Psi^{1/3}_2$, thus one can choose $\mr\Om\propto\Psi^{-1/3}_2$.
(One must keep in mind though that this choice represents an explicit breaking of conformal invariance, 
since the Bianchi identities are not conformally invariant.)
An example of this is the Kerr-(A)dS spacetime.
It is also satisfied for type D conformal structures with a background Maxwell field 
whose PNDs are aligned to the gravitational ones, namely the Maxwell field has the form
$-2\phi o_{(A}\iota_{B)}$; in this case Maxwell equations imply that 
$\mf{f}_a=\phi^{-1/2}\p_a\phi^{1/2}$ and therefore $\mr\Om\propto\phi^{-1/2}$. 
This is the situation for example in the Kerr-Newman-(A)dS spacetime.

Now, for those situations where we have a non-trivial solution to $\C_a\mr\Om=0$ 
(such as the examples mentioned above), we can set $h^{-}=\mr\Om\vp_A\iota^A$ 
and $h^{+}=\mr\Om\vp_Ao^A$, then these fields have the correct $p$ and $w$ weights 
and solve the Teukolsky equations as long as $\vp_A$ solves the LH Dirac equation. 
We then have that each of the fields $\phi^{-}_{A'}$ and $\phi^{+}_{A'}$ defined above 
is a solution to the RH Dirac equation. But a straightforward calculation shows that
\begin{equation}\label{diffD}
 \phi^{-}_{A'}-\phi^{+}_{A'}=-\mr\Om\c_{A'}{}^{A}\vp_{A},
\end{equation}
thus, the two fields are actually the same as long as $\vp_A$ is a LH Dirac field.
That is, both symmetry operators coincide.

For RH Maxwell fields, we take $h^{-}\in\G(\mc{O}(-2)[-1])$, $h^{+}\in\G(\mc{O}(2)[1])$ 
and $h^{0}\in\G(\mc{O}(0)[0])$, and define the following vector potentials:
$A^{-}_{AA'}=o_Ao^B\C_{BA'}h^{-}$, $A^{+}_{AA'}=\iota_A\iota^B\C_{BA'}h^{+}$,
$A^{0}_{AA'}=o_A\iota^B\C_{BA'}h^{0}$, and $\tilde{A}^{0}_{AA'}=\iota_Ao^B\C_{BA'}h^{0}$.
These are all variants of the vector potential in (\ref{LHM2})-(\ref{RHM2}).
We first note that $A^0_{AA'}-\tilde{A}^0_{AA'}=\p_{AA'}h^0$, so these two potentials differ by 
a gauge transformation and we can consider only one of them, say $A^0_{AA'}$.
Now, the associated 2-forms $F^{-}_{ab}=2\c_{[a}A^{-}_{b]}$, etc. decompose into their 
SD (or RH) and ASD (or LH) parts, $\phi^{-}_{A'B'}$, $\psi^{-}_{AB}$, etc.
A similar calculation to the one in eqs. (\ref{LHM2})-(\ref{RHM2}) shows that
the LH parts $\psi^{0,\pm}_{AB}$ are zero if $h^{0,\pm}$ satisfy the corresponding 
wave equations; in all these cases the RH parts consequently satisfy Maxwell equations. 
For the construction of symmetry operators,
we can obtain solutions of the Teukolsky and Fackerell-Ipser equations starting from 
a LH Maxwell field $\vp_{AB}$, by setting $h^{-}=\mr\Om^2\vp_{AB}\iota^A\iota^B$,
$h^{+}=\mr\Om^2\vp_{AB}o^Ao^B$ and $h^{0}=\mr\Om^2\vp_{AB}o^A\iota^B$, 
where, as before, $\C_a\mr\Om=0$. But it is not difficult to show that 
\begin{align}
 & A^{-}_{AA'}-A^{+}_{AA'}=-\mr\Om^2(o_A\iota^C+\iota_Ao^C)\c_{A'}{}^{B}\vp_{BC}
 +\p_{AA'}(\mr\Om^2\vp_{BC}o^B\iota^{C}) \label{diffM1} \\
 & A^{-}_{AA'}-A^{0}_{AA'}=-\mr\Om^2 o_A\iota^C\c_{A'}{}^{B}\vp_{BC} \label{diffM2}
\end{align}
and similarly for the other possible combinations. Thus we see that, as long as $\vp_{AB}$ 
is a LH Maxwell field, all vector potentials differ by gauge transformations, hence they define 
the same RH Maxwell field $\phi_{A'B'}\equiv\phi^{0,\pm}_{A'B'}$, i.e. all symmetry 
operators coincide.

\begin{remark}\label{rem-ef}
The case in which $\C_a\mr\Om=0$ admits non-trivial solutions, namely 
$\mf{f}_a$ is an exact form, has close relations with the usual concept of 
hidden symmetries in General Relativity. First, the fact that $\mf{f}_a$ is closed implies 
that the Weyl connection $\slashed{\c}{}_{a}$ is actually the Levi-Civita connection 
of some metric in the conformal class, say $\slashed{g}{}_{ab}$. 
Furthermore, if $o^A$ and $\iota^A$ are both SFRs, then one can show 
(see \cite[eq. (2.20)]{Araneda2018}) that the almost-complex structure $J$ is parallel for 
$\slashed{\c}$, i.e. $\slashed{\c}{}_{a}J_{b}{}^{c}=0$, so $\slashed{g}{}_{ab}$ is a 
K\"ahler metric. Additionally, using that $\C_a(o^B\iota^C)=0$ and $\mf{f}_a=\mr\Om\p_a\mr\Om^{-1}$
it is easy to verify that the spinor field $X^{AB}:=\mr\Om o^{(A}\iota^{B)}$ satisfies 
$\c_{A'}{}^{(A}X^{BC)}=0$, namely it is a Killing spinor\footnote{The equation 
$\c_{A'}{}^{(A}X^{BC)}=0$ is conformally invariant, so $\c_{AA'}$ here is any Levi-Civita 
connection in the conformal class.}. 
Therefore, the connection (\ref{C}) somehow encodes the conformally K\"ahler structure and the 
existence of Killing spinors in all spacetimes where $\mf{f}_a$ is exact 
(which includes for example the Kerr-(A)dS and Kerr-Newman-(A)dS solutions). 
For a thorough analysis of conformally K\"ahler structures in 4 dimensions, we refer
the reader to \cite{Dun09}.
\end{remark}

\section{Final comments}\label{sec-FC}

The methods and ideas of twistor theory have proven to be extremely useful in a wide range of 
topics in theoretical and mathematical physics, such as string theory and scattering amplitudes, 
loop quantum gravity, integrable systems, quasi-local constructions of mass and angular 
momentum, etc. (see \cite{Ati17} for a recent review, and also references therein).
In this work we have argued that twistor structures are also present in 
perturbation theory of algebraically special spaces, by showing that the standard formalisms 
known in the literature (such as Teukolsky equations) have a geometric structure that is naturally 
interpreted in terms of a 2-dimensional twistor manifold.

The standard definition of (projective) twistor space is as the moduli space of certain complex 
2-dimensional (namely $\a$- or $\b$-) surfaces in a spacetime, and the requirement that this 
space be {\em three}-complex dimensional forces the conformal curvature to be SD or ASD.
We have studied geometric structures induced in a (conformal) spacetime 
by requiring instead the existence of a {\em two}-dimensional twistor manifold, and the relation 
of these structures with the description of linear massless fields propagating in the spacetime.
Our results are valid for conformal structures that {\em are not} necessarily SD or ASD,
but admit a null geodesic congruence that is shear-free (referred to as SFR along the text), 
and we have also assumed that the corresponding spinor field is a two-fold principal null 
direction of the SD curvature.
As mentioned, our main motivation for studying this problem was the recent result 
\cite{Araneda2018} that the Teukolsky equations (that are central to the black hole 
stability problem) are intimately related to a combination 
of conformal, complex and spinor geometry, which is a natural territory of twistor theory.

We have proceeded by following closely the standard constructions in twistor theory, 
adapted to our context where we have a 2D twistor space $\T$.
We showed that $\T$ induces in a natural way a connection $\C_a$ (given by eq. (\ref{C})) 
on spinor bundles in a conformal structure (lemma \ref{lem-IC}). 
We studied the curvature of this connection in section \ref{sec-curv}, which allowed us to 
show that the connection naturally induced on $\b$-surfaces is flat (lemma \ref{lem-Cbs}), 
and which also allowed us to construct line bundles over $\T$ since the 
integrability conditions are satisfied (section \ref{sec-lb}).
We have shown that, in particular, these constructions are intimately related to perturbation 
theory of black hole spacetimes, since the differential operators induced by $\T$ are 
closely associated to Teukolsky operators; see lemma \ref{mthm} and
eqs. (\ref{teuk1}), (\ref{teuk2}), (\ref{FI}).
Furthermore, we showed that sections of line bundles over $\T$ are automatically 
solutions of the Teukolsky equations for massless free fields, and this construction 
resembles the one associated to the Penrose transform, see remarks \ref{rem-teq}
and \ref{rem-we}. Likewise,
our construction of massless free fields with higher spin, that we did in section \ref{sec-mff}, 
gives formulas which are also reminiscent of the ones corresponding to the Penrose transform, 
see remarks \ref{rem-LHF} and \ref{rem-RHD}.
The special case in which we have two 2D twistor spaces, $\T$ and $\T'$, was considered in 
section \ref{sec-2TM} (this case includes for example the Kerr-(A)dS and 
Kerr-Newman-(A)dS solutions).
There, we gave an interpretation of the almost-complex structure (\ref{J}) 
(that is crucial for the construction of the appropriate connection) as induced 
naturally by the complex-manifold structure of the foliations by $\b$-surfaces; and we 
discussed an appropriate notion of holomorphic structures.
It was also shown how to generate solutions to the Teukolsky equations with opposite 
spin weight (and also to the Fackerell-Ipser equation) from line bundles over $\T$ and $\T'$.
We also showed that the different formulas for RH massless free fields (for a given spin), obtained 
from symmetry operators, are actually the same, see eqs. (\ref{diffD}), (\ref{diffM1}), (\ref{diffM2}).

For the case of gravitational perturbations of a curved spacetime, we have shown in section \ref{sec-gp} 
that linearized metric perturbations generated by Hertz potentials still possess a 2-dimensional twistor 
manifold to linear order, by proving that the corresponding linearized ASD curvature spinor is algebraically 
special (lemma \ref{lem-as}) and that this implies that the background SFR continues to be an SFR 
at the linear level (lemma \ref{lem-lSFR}, which is, as emphasized, a linearized version of the 
Goldberg-Sachs theorem in one direction).
This result and the ones mentioned before  suggest some possible research lines that we believe 
deserve further investigation \cite{Ara19}.
First, it would be interesting to understand the relationship between the 2D twistor space considered 
in this work
and the construction of an asymptotic twistor space\footnote{I am grateful to L. Mason for this suggestion.},
whose properties (such as the fact that it is an Einstein-K\"ahler manifold) 
might induce some interesting structures in the spacetime.
Second, we note that, in twistor theory, the treatment of the gravitational field is through consideration
of deformations of the complex structure of twistor space. 
At present it is not clear to us if some form of such procedures 
could also be applied to our case, and, even if so, whether it could lead to a better understanding 
of the structure of linearized gravity on curved spacetimes.
However, since this procedure involves the introduction of an infinity twistor that breaks conformal 
invariance, it might be interesting to see if a similar mechanism can be applied in our formalism 
in order to also break conformal invariance for the treatment of gravitational perturbations.
Finally, as already stated in remark \ref{rem-l2dts}, the results of section \ref{sec-gp}
suggest that there might be a close relationship between the constructions studied in this work
and the theory of Hyperheaven spaces \cite{Ple76, Fin76, Boy79}.
Results about all these problems will be presented elsewhere \cite{Ara19}.
In any case, in view of the techniques used in some recent very important results concerning 
the classical problems in mathematical Relativity (in particular see \cite{And19}), the 
application of spinor and twistor methods to these problems does not seem to be exhausted.

\subsection*{Acknowledgements}

It is a pleasure to thank Steffen Aksteiner, Lars Andersson, Thomas B\"ackdahl, Igor Khavkine and
Lionel Mason for very helpful discussions, that took place at the Institut Mittag-Leffler 
(Djursholm, Sweden) in the fall 2019.
I am also very thankful to Tim Adamo, Maciej Dunajski and George Sparling for comments 
about this work during the conference ``Twistors meet Loops in Marseille'', held at CIRM (France) 
in September 2019; 
in particular I want to thank M. Dunajski for several illuminating conversations in this conference 
and also during a visit to Cambridge University in November 2019.
The hospitality and support of all the institutions mentioned above are also gratefully acknowledged.
Finally I thank Gustavo Dotti, Jos\'e Luis Jaramillo, Oscar Reula and Juan Valiente Kroon 
for supportive comments on this work and on a previous version of this manuscript.
This work is partially supported by a postdoctoral fellowship from CONICET (Argentina).

\end{document}